\documentclass{amsart}
\usepackage{enumerate}
\usepackage{amsmath}
\usepackage{amsthm}
\usepackage{amsfonts}
\usepackage{amssymb}
\usepackage{mathrsfs}
\usepackage{stmaryrd}
\usepackage{tikz} 
\usetikzlibrary{matrix,arrows,decorations.pathmorphing}
\usepackage[pdfpagescrop={90 60 520 725}, bookmarks, colorlinks, breaklinks]{hyperref}  
\hypersetup{linkcolor=blue,citecolor=blue,filecolor=black,urlcolor=blue}

\newtheorem{theorem}{Theorem}[section]
\newtheorem{proposition}[theorem]{Proposition}

\newtheorem{lemma}[theorem]{Lemma}

\theoremstyle{remark}

\theoremstyle{plain}

\newcommand\wt\widetilde
\newcommand\sheaf\mathcal
\newcommand\complex\mathscr
\newcommand{\exterior}[1]{\ensuremath{{\textstyle\bigwedge}^{\! #1 }}}
\newcommand\lto\longrightarrow

\newcommand{\Sp}[1]{\ensuremath{\text{Splittings} #1  }}
\newcommand{\uSp}[1]{\ensuremath{\underline{\text{Splittings}}} #1  }

\newcommand\Isom{\ensuremath{\text{Isom}}}
\newcommand\Conn{\ensuremath{\text{Conn}}}
\newcommand\Ext{\ensuremath{\text{Ext}}}
\newcommand\Sym{\ensuremath{\text{Sym}}}
\newcommand\Hom{\ensuremath{\text{Hom}}}

\newcommand\End{\ensuremath{\text{End}}}

\newcommand\Spec{\ensuremath{\text{Spec}}}
\newcommand\Defo{\ensuremath{\text{Defo}}}
\newcommand\Maps{\ensuremath{\text{Maps}}}

\newcommand\h{h^{12}}

\newcommand\CC{\ensuremath{\mathbb C}}

\newcommand\ZZ{\ensuremath{\mathbb Z}}
\newcommand\RR{\ensuremath{\mathbb R}}
\renewcommand\AA{\ensuremath{\mathbb A}}
\newcommand\X{{\ensuremath{X}}}
\newcommand\M{{\ensuremath{\mathcal{M}}}}
\newcommand\D{{\ensuremath{\mathcal{D}}}}
\newcommand\sM{{\ensuremath{\mathcal{SM}}}}
\newcommand\SM{{\ensuremath{\mathfrak{M}}}}

\renewcommand\S{{\ensuremath{\mathcal{S}}}}
\newcommand\OO{{\ensuremath{\mathcal{O}}}}
\def\d{{\mathrm d}}
\def\L{{\mathcal L}}
\def\Res{{\mathrm{Res}}}
\def\Z{{\Bbb Z}}
\def\Tr{{\mathrm{Tr}}}
\def\C{{\Bbb C}}
\font\teneurm=eurm10 \font\seveneurm=eurm7  \font\fiveeurm=eurm5
\newfam\eurmfam
\textfont\eurmfam=\teneurm \scriptfont\eurmfam=\seveneurm
\scriptscriptfont\eurmfam=\fiveeurm

\font\teneusm=eusm10 \font\seveneusm=eusm7 \font\fiveeusm=eusm5
\newfam\eusmfam
\textfont\eusmfam=\teneusm \scriptfont\eusmfam=\seveneusm
\scriptscriptfont\eusmfam=\fiveeusm

\font\tencmmib=cmmib10 \skewchar\tencmmib='177
\font\sevencmmib=cmmib7 \skewchar\sevencmmib='177
\font\fivecmmib=cmmib5 \skewchar\fivecmmib='177
\newfam\cmmibfam
\textfont\cmmibfam=\tencmmib \scriptfont\cmmibfam=\sevencmmib
\scriptscriptfont\cmmibfam=\fivecmmib

\def\red{{\mathrm{red}}}
\def\Res{{\mathrm{Res}}}
\numberwithin{equation}{section} 
\begin{document}
\bibliographystyle{hep}

\parindent=0cm
\parskip=0.5cm

\title[Super Atiyah Classes]{Super Atiyah Classes and Obstructions to Splitting of Supermoduli Space}
\author{Ron Donagi}
\address{Department of Mathematics, University of Pennsylvania, Philadelphia, PA 19104}
\email{donagi@math.upenn.edu}
\author{Edward Witten}
\address{Institute for Avanced Study, Princeton NJ  08540}
\email{witten@ias.edu}

\date{\today}

\begin{abstract}
The first obstruction to splitting a supermanifold $S$ is one of the three components of its super Atiyah class, the two other components being the ordinary Atiyah 
classes on the reduced space $M$ of the even and odd tangent bundles of $S$. We evaluate these classes explicitly for the moduli space of super Riemann surfaces 
(``super moduli space'') and its reduced space, the moduli space of spin curves. These classes are interpreted in terms of certain extensions arising from line bundles on the 
square of the varying (super) Riemann surface. These results are used to give a new proof of the non-projectedness of $\SM_{g,1}$, the moduli space of super Riemann surfaces with one puncture.
\end{abstract}
\maketitle
\newpage
\tableofcontents
\thispagestyle{empty}
\newpage
\section{Introduction}
It was shown in  \cite{Not projected} that various natural moduli spaces of super Riemann surfaces, with and without NS punctures, are not projected. The argument boiled down to showing that the  obstruction class $\omega$ to the splitting of these moduli spaces (by which we will always mean the {\em first} obstruction class, denoted  $\omega_2 $ in \cite{Not projected})  does not vanish. In the present work we interpret this  first obstruction class, for an arbitrary supermanifold, as a component of the  super Atiyah class. There are two other components, which are the ordinary (bosonic) Atiyah classes of the even and odd tangent bundles. For the relevant moduli spaces, we are able to write down explicit formulas for the obstruction class and the Atiyah class. This leads to an alternate proof for the non-vanishing of $\omega$ for the moduli space $\SM_{g,1}$ of super Riemann surfaces of genus $g \geq 2$ with one  (Neveu-Schwarz) puncture. This was an important input in the other 
major results of   \cite{Not projected}, including non-projectedness of  $\SM_{g}$ itself in genus $g \geq  5$. It would be desirable to have an explicit formula also for the super Atiyah class; this should yield the results for the obstruction class and the two Atiyah classes as special cases. This remains an open problem.

As we will review below,  the first obstruction to splitting the supermanifold $S=(M,\OO_S)$ 
is given \cite{Not projected, Manin} by a cohomology class
\[
 \omega = \omega _2 \in H^1 (M, (T_{+} \otimes \exterior{2}    {T_-}  ^{\vee}))
=\Ext^1(\exterior{2}   {T_-}   ,  T_{+} ),
\]
where $T_{\pm}$ denote the even and odd tangent bundles of $S$.
We start in section \ref{2} by interpreting the obstruction class $\omega_S$ to the splitting of a supermanifold $S$ in terms of various data associated to $S$. One interpretation is based on rings of differential operators: we interpret  $\omega_S$ as an extension class encoded in the rings $D^i_{S,-}$ of $i$-th order differential operators on $S$ whose symbol is purely odd.
More precisely, the relevant ring is $\bar{D} := D^2_{S,-} / D^1_{S,-}$. It is an extension of $\exterior{2}T_-$ by $T_+$, so it has an extension class in  $\Ext^1(\exterior{2}   {T_-},  T_{+} )$. We check that this coincides with the first obstruction class $\omega$.
A second interpretation is in terms of second order deformations in $S$ of points of the underlying manifold $M$. The latter form an affine bundle (over a particular vector bundle over $M$), and  the class of this affine bundle
can be identified with $\omega$ .

The Atiyah class $\alpha_{M,V} = \alpha_V$ of a complex vector bundle $V$ on a manifold $M$ measures the obstruction to existence of a global connection on $V$.  It lives in the cohomology group $H^1(M,T^{\vee} \otimes End(V))$,
since a connection is a section of a torsor (=principal homogeneous bundle) under the vector bundle of $1$-forms with values in $End(V)$.
The Atiyah class can be viewed as a non-commutative Chern character. Indeed, the Chern character is recovered as the trace: $ch(V) = \Tr(\exp{(\alpha_V)})$, meaning that
the $i$-th component of the Chern character is recovered as: $ch_i(V) = \Tr(\frac{(\alpha_V)^i}{i!})$.
The most important case is when $V$ happens to be the tangent bundle of $M$. The notation is then abbreviated to $\alpha_{M,T} = \alpha_{M}$, and some simplifications occur, related to the notion of torsion free connection, which makes sense for $T$ but not for a general $V$.
We note in section \ref{2.3} that this Atiyah class has characterizations, in terms of rings of differential operators and second order deformations of points, that are very reminiscent of the corresponding characterizations of the obstruction class. 

The mystery is resolved in the next section \ref{cghjjlkh}: there are straightforward super-versions $\alpha_{S,V}$  and $\alpha_{S}$  of the Atiyah class of a bundle (or the tangent bundle) on a supermanifold $S$. When the latter is restricted back to the reduced manifold $M$, it splits in three pieces. Two of these are the ordinary Atiyah classes of the even and odd tangent bundles of $S$. The third is the obstruction class.

In section \ref{3} we focus on moduli spaces of super Riemann surfaces, and give an algebro-geometric interpretation of their obstruction classes. We have seen in section \ref{2}  that on any supermanifold, the obstruction class $\omega$ is the extension class of a vector bundle on $S$, with quotient $\exterior{2}   {T_-}$ and subspace $T_+$.
When $S$ is the moduli space $\SM_g$ of super Riemann surfaces, these bundles have natural cohomological interpretations, in terms of spaces of differentials on the variable curve. Our point is that the extension class too has a natural cohomological interpretation.
This is a little easier to state in terms of the duals: a point of $\SM_g$ is a super Riemann surface, which determines a Riemann surface $C$ and a square root $T_{C,-} = (T_C)^{1/2}$ of its tangent bundle. The fiber at that point of the  dual $T_{ {\SM_g}  ,+}^{\vee}$ is the space $H^0(C, K_C^2)$ of quadratic differentials on $C$, while the fiber of $T_{ {\SM_g}  ,-}^{\vee}$ is the space $H^0(C, K_C^{3/2})$ of sesquilinear differentials on $C$, and the fiber of  $S^2 T_{ {\SM_g}  ,-}^{\vee}$ is $S^2  H^0(C, K_C^{3/2})$.
The latter can be interpreted as the symmetric part (with respect to the involution that fixes the diagonal) of
$ \otimes^2 H^0(C, K_C^{3/2}) = H^0(C \times C, K_C^{3/2} \boxtimes  K_C^{3/2})$. We show that there is a natural line bundle $\OO(3,3,1)$ on $C \times C$ that fits into a short exact sequence:
\[
0 \to \OO(3,3,0) \to \OO(3,3,1) \stackrel {\Res}{\to}  K_C^2 \to 0,
\]
where $\OO(3,3,0) $ is the invertible sheaf on  $C \times C$ of sections of the line bundle 
$K_C^{3/2} \boxtimes  K_C^{3/2}$,  $\OO(3,3,1) $ is the invertible sheaf on  $C \times C$ of sections of the same line bundle $K_C^{3/2} \boxtimes  K_C^{3/2}$ that are allowed to have a first order pole along the diagonal, and $\Res$ is the restriction to the diagonal. These three sheaves have no higher cohomology, so the short exact sequence of sheaves gives a short exact sequence of spaces of global sections:
\[
0 \to H^0(  C \times C , K_C^{3/2} \boxtimes  K_C^{3/2} ) \to H^0( C \times C  ,\OO(3,3,1)) \stackrel {\Res}{\to}  H^0( C, K_C^2) \to 0,
\]

The involution of $C \times C$ acts on everything, and the invariant part gives the extension whose class is the obstruction class $\omega_{\SM_g}$. There is an analogous algebro-geometric interpretation of the  Atiyah class $\alpha_{\M_g}$ of the moduli space of ordinary Riemann surfaces. The proof is harder in this case, so we give it first, in section \ref{oyf}, and then give the parallel but easier proof  for the obstruction class in section \ref{sfg}. It is not clear whether these fit together to give an explicit description of the super 
 Atiyah class $\alpha_{\SM_g}$ of the moduli space of super Riemann surfaces.

In section \ref{4} we use this algebro-geometric interpretation of  the obstruction class of $\SM_g$, and its variant for super Riemann surfaces with a puncture,  to give a new proof of the non-projectedness of the moduli space $\SM_{g,1}$ of super Riemann surfaces with a puncture, for genus $g \geq 2$.

\def\X{{\mathcal X}}
\subsection{Orbifolds and Stacks} \label{vyhdfghvb}

We generally assume here a familiarity with the basics of super manifolds and super Riemann surfaces.   The background needed for our
purposes is explained in \cite{Not projected}, and in other references cited there.  In this paper, a super manifold is assumed to be defined over a field of characteristic
zero, which we will eventually take to be the complex numbers. 

One detail should perhaps be explained here. (See section 3.2.1 of \cite{Not projected} for a fuller explanation, and \cite{ER,Ra} for the relevant facts
about $\S\M_g$.)  A generic Riemann surface $\Sigma$ of genus $g\geq 3$ has no automorphisms. Because automorphisms
can occur, the moduli ``space'' of Riemann surfaces is best understood as an orbifold or a stack.  However, these  issues are more prominent
for super Riemann surfaces, because every super Riemann surface is infinitesimally close to a super Riemann surface with a $\Z_2$ group
of automorphisms.  Indeed, every split super Riemann surface $S=\Pi T^{1/2}C = S(C,T^{1/2})$ (where $C$ is an ordinary Riemann surface) has such an automorphism group, acting as $\{\pm 1\}$ on the fibers of the fibration $S(C,T^{1/2}) \to C$.  

Actually,
a spin bundle of $C$, understood as a line bundle $K^{1/2}\to C$ with an isomorphism $K^{1/2}\otimes K^{1/2}\cong K$, has a natural group
$\{\pm 1\}$ of automorphisms.  As a result, the spin moduli space $\S\M_g$ should be understood (even if we ignore the possibility that $\Sigma$
itself might have automorphisms) as an orbifold or stack with $\Z_2$ automorphism group and trivial action of $\Z_2$.  $\S\M_g$ parametrizes a family
$\X_g\to\S_g$
of genus $g$ curves $C$ each of which is endowed with a $K^{1/2}$ up to isomorphism, but there is no line bundle over $\X_g$ whose restriction to
each fiber is isomorphic to $K^{1/2}$.  Accordingly, the odd tangent bundle $T_-\SM_g\to \S\M_g$ of the super moduli space $\SM_g$, whose fiber at a point in 
$\S\M_g$ corresponding to $C$ is supposed to be $H^1(C,T^{1/2})$ where $T^{1/2}\cong K^{-1/2}$, does not exist as an ordinary vector bundle over $\S\M_g$.  It exists as a
$\Z_2$-twisted vector bundle, in other words a vector bundle twisted by a $\Z_2$ gerbe.

The reason that this need not concern us is that what enters the definition of the 
(first) obstruction class to splitting of $\SM_g$  is
not $T_-\SM_g$ but its second exterior power $\wedge^2 T_-\SM_g$.  This exists as an ordinary vector bundle over $\S\M_g$.  The obstruction
class involves an extension by $\wedge^2 T_-SM_g$ of $T_+\SM_g\to \S\M_g$, which is also an ordinary vector bundle (with fiber $H^1(C,T)$) .

Concretely, although there is no line bundle $\L\to \X_g$ that restricts to $K^{1/2}$ (up to isomorphism) on each fiber of $\X_g\to\S\M_g$, there is a line bundle
$\mathcal R \to \X_g\times_{\S\M_g}\X_g$ that restricts  to $K^{1/2}\boxtimes K^{1/2}$ on each fiber of $\X_g\times_{\S\M_g}\X_g\to \S\M_g$. 
(The fibers of $\X_g\times_{\S\M_g}\X_g\to \S\M_g$ are copies of $C\times C$.)
 The reason for this is simply that
the automorphism group $\{\pm 1\}$ of $K^{1/2}$ acts trivially on $K^{1/2}\boxtimes K^{1/2}$, so $\mathcal R$ exists locally up to unique isomorphism
and there is no problem to construct it globally.   It is convenient to simply denote $\mathcal R$
as $K^{1/2}\boxtimes K^{1/2}$, and we will do so in the text.   $\wedge^2 T_-\SM_g$ is the part of $H^2(C\times C,K^{1/2}\boxtimes K^{1/2})$
that is odd under the involution that exchanges the two factors of $C\times C$.

\section{Obstruction and Atiyah classes}\label{2}

\subsection{Obstruction class via differential operators}

\mbox {   }

Start with a supermanifold $S=(M,\OO_S)$, 
where $M$ is a manifold  and the sheaf of rings $\OO_S$ on $M$
is locally isomorphic to an exterior algebra over a vector bundle $V^\vee$ on $M$. 
The vector bundle $V$ can actually be recovered, globally, using 
the maximal nilpotent ideal $J \subset \OO_S$,
which is the sheaf of ideals generated by all elements of odd degree. 
The quotient $\OO_S / J$ is identified with $\OO_M$,
so $M$ is always a sub(super)manifold of $S$.
The dual of the bundle $V$ is recovered globally as the quotient $V^{\vee} := J/J^2$. 
The associated graded sheaf of rings of  $\OO_S$ with respect to $J$ 
is the  exterior algebra on $V^{\vee}$. 
It is the structure sheaf of a supermanifold $S(M,V)$, 
and we say that $S$ is split if it is isomorphic to $S(M,V)$. 

The first obstruction to splitting the supermanifold $S=(M,\OO_S)$,
or to finding a projection $S \to M$, 
is given \cite{ Manin, Not projected}  by a cohomology class
\[
 \omega = \omega _2 \in H^1 (M, (T_{+} \otimes \exterior{2}    {T_-}  ^{\vee}))
=\Ext^1(\exterior{2}   {T_-}   ,  T_{+} ),
\]
where $T_{\pm}$ denote the even and odd tangent bundles of $S$. (These are vector bundles on $M$, in fact $T_{+}$ is just the tangent bundle of $M$, and $T_{-}$ is the above $V$.) 
One way to see this is to note that the parity-respecting homomorphisms $\OO_M \to \OO_S/J^3$ form a torsor under the derivations of $\OO_M$ with values in $J^2/J^3 = \exterior{2}    {T_-}  ^{\vee} $.

Any such class $\omega$ can be interpreted as the extension class ${[} D_{\omega} {]} $  
of an extension of vector bundles on $M$:
\begin{equation} \label{D2}
0 \to  T_{+} \to D_{\omega} \to \exterior{2}   {T_-}    \to 0.
\end{equation}

In particular, this applies to  our first obstruction class $\omega$. 
We claim that the extension $D_{\omega}$ corresponding to  $\omega$ can be realized as 
a certain sheaf of differential operators along $S$. 
Let ${\D}_S$ denote the sheaf of differential operators on $S$,
and ${{\D}_S}_{|M}$ its restriction to $M$. 
Technically, both can be viewed as sheaves on $M$: 
${\D}_S$ is a sheaf of $\OO_S$ modules, while 
${{\D}_S}_{|M}$ is obtained by tensoring over $\OO_S$ with $\OO_M=\OO_S/J$,
i.e. dividing by $J$, so it is a sheaf of $\OO_M$ modules.
But it is more suggestive to say that ${\D}_S$ is a sheaf on $S$
while ${{\D}_S}_{|M}$ is a sheaf on $M$.
We can always identify the sheaf $\D_M$ of differential operators on $M$ as a subsheaf of ${{\D}_S}_{|M}$: 
if $A$ is a differential operator on $M$, 
its action on a function $f$ on $S$ is defined to be $A(f_{|M})$.
When $S=S(M,V)$ is split, we can also identify $\exterior{\bullet}   {V}$ as a subsheaf of  ${{\D}_S}_{|M}$,
acting by contractions. Putting these together, we can identify:
\begin{equation}\label{Scjghi}
{{\D}_{S(M,V)}} _{|M} =  \exterior{\bullet}   {V} \otimes _{\OO_M}  \D_M. 
\end{equation}

Let ${\D}^i_S$ denote the subsheaf of ${\D}_S$ consisting of differential operators of order $\leq i$ on $S$,
and let ${D}^i={D}^i_S$ denote its restriction to $M$, 
the sheaf on $M$ of differential operators of order $\leq i$ along $S$,
a subsheaf of ${{\D}_S}_{|M}$.
(The order of a differential operator is the maximal number of differentiations involved, 
with respect to both the even and the odd variables.)
Let ${D}_{-}^i={D}_{S,-}^i$ denote  
the subsheaf of differential operators of order $\leq i$ on $M$ along  $S$ 
whose symbol, or $i$-th order term, is purely odd, 
i.e. it is given in terms of local coordinates $x,\theta$ as a sum of terms of the form 
$f(x)\partial / \partial \theta_{k_1} \cdots \partial / \partial \theta_{k_i}$, 
with no $\partial / \partial {x_j}$ involved. 
Note that this is not the restriction to $M$ of any reasonable sheaf ${\D}_{S,-}^i$ on $S$: 
if we did not insist on restricting to $M$, we would allow more general expressions of the form
$f(x,\theta)\partial / \partial \theta_{k_1} \cdots \partial / \partial \theta_{k_i}$. 
When we change the choice of the local coordinates $x,\theta$, 
such terms can go to terms whose symbols involve the forbidden
$ \partial / \partial x$, but always with a coefficient in the nilpotent ideal $J$. 
The effect of restricting to $M$ is to make the definition of ${D}_{-}^i$ independent of 
the coordinates.
(We check the independence explicitly in \eqref{yfouyfo} below.)
The precise claim is that our extension $D_{\omega}$ can be identified as the sheaf
\begin{equation} \label{fuytfyrdyvdr}
\bar{D} := {D}_{-}^2/{D}_{-}^1
\end{equation}
 on $M$.
In local coordinates,  $  {D}_{-}^2$ is generated by terms written schematically as $1, \partial / \partial x, \partial / \partial \theta,    \partial ^2/ \partial \theta^2$. Similarly, $  {D}_{-}^1$ is generated by terms written schematically as $1, \partial / \partial  \theta$, so the quotient $\bar{D}$ is generated by terms written  as
$ \partial / \partial x,   \partial ^2/ \partial \theta^2$ and it defines an extension as in (\ref{D2}).

We could prove this claim directly by tracing through Manin's construction  \cite{Manin} of the class $\omega$. 
Instead, we will start with the description of $\omega$ in  \cite{Not projected}  as the obstruction to splitting the second order neighborhood of $M$ in $S$.

Recall that a supermanifold $S=(M,\OO_S)$ admits a natural increasing filtration
\[
S_{red} = S^{(0)} \subset S^{(1)} \subset \dots  \subset S^{(i-1)} \subset S^{(i)} \subset \dots \subset S^{(n)}  =S.
\] 
The $S^{(i)}$ are defined in terms of the
nilpotent ideal sheaf $J\subset \OO_S$:
\begin{equation}\label{Si}
S^{(i)} := (M, \OO_{S^{(i)}} =\OO_S/J^{i+1}).
\end{equation}
These $S^{(i)}$ are locally ringed subspaces of $S$, 
though they are not supermanifolds, except for the extremes $i=0,n$.  
(They are superanalogs of everywhere non-reduced schemes in ordinary algebraic geometry. ) 
Nevertheless, their sheaves of differential operators are well defined, and satisfy:
\begin{equation} \label{oiut}
{D}^i_{S^{(i)} }= {D}^i_S.
\end{equation}

By a {\em splitting} of the supermanifold $S$ we mean 
an isomorphism from the split supermanifold $S(M,V)$ to $S$
that induces the identity on both the underlying reduced space $M$ and the odd tangent bundle $V$.
The family of all splittings of $S$ is parametrized by 
$$\Sp(S) := \Isom_{M,V}(S(M,V),S).$$
This is a torsor over the group $\Sp(S(M,V))$.
A bit more generally, we have the notion of a splitting of  the superspace $S^{(i)}$, i.e. 
an isomorphism from $S(M,V)^{(i)}$ to $S^{(i)}$
that induces the identity on both  $M$ and $V$, and
the parameter space  $\Sp(S^{(i)})$ of all  such splittings, which is 
a torsor over the group \linebreak $\Sp(S(M,V)^{(i)})$.
Working locally over $M$, we get sheaves of groups 
\linebreak $\uSp(S(M,V)^{(i)})$ and torsors $\uSp(S^{(i)})$ over them. 
In particular for $i=2$, the sheaf of groups 
$\uSp(S(M,V)^{(2)})$ can be identified with the vector bundle
$T_{+} \otimes \exterior{2}    {T_-}  ^{\vee}$, while the torsor over it 
 $\uSp(S^{(2)})$ becomes an affine bundle modeled on 
$T_{+} \otimes \exterior{2}    {T_-}  ^{\vee}$,
the structure group being the group of translations.
The first obstruction $\omega_2$ to the splitting of $S$ was identified in \cite{Not projected} with the class of this affine bundle, or equivalently with the extension class of the sequence of vector bundles:
\begin{equation}\label{rdge}
0 \to T_{+} \otimes \exterior{2}    {T_-}  ^{\vee} \to X \stackrel {\pi} {\to} \OO_M \to 0
\end{equation}
from which $\uSp(S^{(2)})$ is recovered as $\pi^{-1}(1)$.

We need to connect extensions \eqref{D2} and \eqref{rdge}. Very generally, giving an extension
\[
0 \to A \otimes C \to X \to B \to 0
\]
is equivalent to  giving an extension
\[
0 \to A \to X' \to B \otimes C^{\vee} \to 0.
\]
Of course, both types of extension are given by elements of 
$\Ext^1(B,A \otimes C)$. Explicitly, we can go back and forth: $X'$ is recovered as
\[
X' := (X \otimes  C^{\vee} ) /(A \otimes \End_0(C)),
\]
where $\End_0$ denotes traceless endomorphisms; and conversely:
\[
X := \ker (X' \otimes C \to B \otimes \End_0(C)).
\]
We apply this with 
$A=T_{+}   $,
$B= \OO_M  $,
$C=  \exterior{2}   {T_{-} } ^{\vee}$,
$X'=\bar{D}$ as in \eqref{fuytfyrdyvdr}, 
and the $X$  in \eqref{rdge}. So we need to identify $X$ with
\[
 \ker (\ \bar{D} \otimes \exterior{2}   {T_{-}} ^{\vee} \to  \End_0(\exterior{2}   {T_{-} } ^{\vee}) \ ).
\]
This boils down to an appropriate map from $\uSp(S^{(2)})$ to 
$$\bar{D} \otimes \exterior{2} {T_{-} } ^{\vee} = \Hom( \exterior{2} {T_{-} }, \bar{D} ).$$ 
And indeed, given a local isomorphism $S(M,V)^{(2)} \to S^{(2)} $, 
we use \eqref{oiut} to identify $D^i_{S(M,V)}$ with $D^i_{S}$ for $i=1,2$, 
and get a corresponding identification of $\bar{D}_{S(M,V)}$ with $\bar{D}_{S}$. But by 
\eqref{Scjghi} there is a natural inclusion of $\exterior{2} {T_{-} }$ into ${{D}_{S(M,V)}}  $, 
which gives the required map from $\exterior{2} {T_{-} }$ to $\bar{D}_{S}$.

\subsection{Obstruction class via map spaces }

\mbox {   }

For a supermanifold $S$ and a super (i.e. $\ZZ/2$-graded)   Artinian $\CC$-algebra $B$, let 
\[
S(B):=  \Maps(\Spec(B),S)
\]
be the set of $B$-points of $S$, i.e. the set of all maps from $\Spec(B)$ to $S$.
This is a manifold, i.e. it is finite dimensional, smooth, and purely even.
Indeed, since the reduced space of $\Spec(B)$ is compact (it is in fact just a finite set of points),
this set of maps has a natural structure of a finite dimensional manifold. 
And since maps of supermanifolds, or superschemes, are by definition always even,
there is nothing odd in the structure sheaf of $S(B)$.
Our main example will be the ring $B_i =\CC[\eta_1, \ldots, \eta_i]$, where the $\eta_j$ are independent odd variables. 
Its spectrum $\Spec(\CC[\eta_1, \ldots, \eta_i]) = \AA^{(0|i)}$ is the odd affine space.
 We think of $S(B_i)$ as the space of $(0|i)$-dimensional deformations in $S$  of points of $M=S_{red}$. 
Clearly  $S(B_1)$ can be identified with the total space of  $T_-M$ 
(considered as an ordinary  vector bundle, not a super  vector bundle).

For $S(B_2)$ the situation is more interesting. 
It is a fiber bundle over $M$, whose fibers are isomorphic to vector spaces; 
but it is not a vector bundle,  only an \emph{affine} bundle. 
In more detail: 
Let $\pi: Y \to M$ be  the projection, where 
$Y$ is the total space of the (even)  vector bundle $T_{-} \oplus  T_{-} $.
There is a natural map (of fibre bundles over $M$):  
\linebreak $p: S(B_2) \to Y$.
This $p$ makes $S(B_2)$ into a $\pi^*T_+$-torsor, i.e.
an affine bundle over  $Y$, 
modeled on  the pullback vector bundle $\pi^*T_+$.

This is easy to see using local coordinates 
$(x| \theta) = (x_1, \dots, x_m | \theta_1, \dots,  \theta_n) $ on $S$. 
A map $\Spec(B_2) \to S$ is specified in terms of  {\em even} parameters $x^0, h, v^1,v^2$ by:
\begin{align*}
        x =& x^0 + \eta_1 \eta_2 h   \\
\theta =& \eta_1 v^1 + \eta_2 v^2,
\end{align*}
If we now change coordinates to:
\begin{align*}
   \tilde{x} =& f(x) + \theta_i \theta_j g^{ij}(x)  + \ldots  \\
 \tilde{\theta} =& \theta_i u^i(x)  + \ldots,
\end{align*}
where the $\theta_i$ are the components of $\theta$, we find the transformation formulas:
\begin{align}
 \notag  \tilde{x}^0 &= f(x^0)   \\
 \notag   \tilde{v}^a &= v^a u(x^0), \quad  a=1,2 \\
   \tilde{h} &= h f'(x^0)  + (v^1_i v^2_j - v^2_i v^1_j )  g^{ij}(x^0)  \label{h}, 
\end{align}
where the $v^a_i$ are the components of $v^a$. 
Since $f',u$ are the transition functions for $T_{\pm}$ respectively, 
we see that: $x^0$ describes a point of $M$; 
$x^0$ together with the two $v^a$  describe a point of  the total space $Y$ of $T_{-} \oplus  T_{-} $ over $M$; 
while $h$ lives in an affine bundle over this which is modeled on $\pi^*T_+$ as claimed.

In general, given a vector bundle $U$ over a space $Y$, the affine bundles on $X$ modeled on $U$ are parametrized by the cohomology group $H^1(Y,U)$. In our case, $Y$ is the  total space of $T_{-} \oplus  T_{-}$ over $M$, and $U$ is the pullback to $Y$ of $T_+$. Therefore, cohomology classes on $Y$ can be replaced by their direct image on $M$ via the projection map $\pi: Y \to M$:
\[
H^i(Y, U)  = H^i(Y, \pi^*T_+)  =  H^i(M, \pi_* \pi^*T_+)  = H^i(M,  \pi_* \OO_Y \otimes     T_+). 
\]
Since $Y$ is a vector bundle over $M$, the direct image sheaf $ \pi_* \OO_Y $ is the tensor algebra over the dual bundle.
So:
\[
H^i(Y, U)   = H^i(M,  \pi_* \OO_Y \otimes     T_+) 
= H^i(M,  \otimes^{\bullet}(T_{-} \oplus  T_{-})^{\vee}  \otimes     T_+).
\] 
In particular, the class of our affine bundle $S(B_2)$ on $Y$ can be pushed forward to a class on $M$, 
and the bilinear dependence of  (\ref{h}) on the $v$'s shows that this extension class appears in $\bullet=2$, 
in fact in the summand $H^1(M, \Hom(   \exterior{2}   {T_-} ,  T_+))$ of the direct image, 
which is the right place for the obstruction class. 
Geometrically, this means that our affine $T_+$-bundle on the  total space 
$Y$ of $T_{-} \oplus  T_{-}$ over $M$ is the pullback via the biliinear map
\[
T_{-} \oplus  T_{-} \to  \exterior{2}   {T_-} 
\]
of an affine $T_+$-bundle on the  total space of $ \exterior{2}   {T_-} $ over $M$. This bundle is still given by the transformation formula (\ref{h}), which is now interpreted as linear in $ \exterior{2}   {T_-} $. But this is immediately recognized as the transformation formula defining the differential operators ring $\bar{D} := {D}_{-}^2/{D}_{-}^1$
 from \eqref{fuytfyrdyvdr} in the previous section. We conclude that the class of the affine bundle involved in $S(B_2)$ and the differential operators extension class $[\bar{D}]$ are actually equal. 
Using the  local coordinates $x, \theta$ on $S$, a point of $S(B_2)$ which deforms a point of $M$ with a given $x^0$ is specified by $v^1, v^2,h$ as above, while $h$ and $v^1 \otimes v^2$ give the components of an element of $\bar{D}$ in $T_+M$ and $\exterior{2}   {T_-}$, respectively, according to \eqref{D2}. We summarize the conclusions of the last two sections in:
\begin{proposition}\label{interpretations of obstruction}
For any supermanifold $S$ with reduced space $M$, the following three classes in $H^1(M, \Hom(   \exterior{2}   {T_-} ,  T_+))$ are equal:
\begin{enumerate}
\item The first obstruction $\omega = \omega_2$ to the splitting of $S$;
\item The extension class $[\bar{D}]$  of the ring of differential operators $\bar{D} := {D}_{-}^2/{D}_{-}^1$;
\item The class of the deformation space $S(B_2)$ as an affine $T_+$-bundle over the total space $Y$ of $T_{-} \oplus  T_{-}$ over $M$. 
\end{enumerate}
\end{proposition}

\subsection{Atiyah class}\label{2.3}

For supermanifolds, our Proposition \ref{interpretations of obstruction} gives the equivalence of three objects.
All three objects, and the proposition, have bosonic versions for a complex manifold $M$ with tangent bundle $T=T_M$, where the obstruction class is replaced by the Atiyah class \cite{At,K}. 

We start with the sheaf $\Conn(T)$ whose sections on an open $U \subset M$ consist of all (holomorphic) connections on $T_U$. This is the sheaf of sections of an affine bundle modeled on $\Hom(   \otimes^2 T, T)$. 
It contains the subsheaf $\Conn_{t\negthinspace f}(T)$ of torsion free connections. This in turn is the sheaf of sections of an affine bundle modeled on $\Hom(   \Sym^2 T, T)$. 
The bosonic analogue of the obstruction class $\omega$ is the {\em Atiyah class} of the tangent bundle,
 $\alpha_T \in H^1(M, \Hom(   \Sym^2 T, T))$, which is the obstruction to finding a global section of $\Conn_{t\negthinspace f}(T)$. (Alternatively, this can also be defined as the obstruction to finding a global section of $\Conn(T)$; 
a global connection determines uniquely its torsion-free part, which is a section of $\Conn_{t\negthinspace f}(T)$.
A priori  the obstruction to finding a global section of $\Conn(T)$ lives in $H^1(M, \Hom(   \otimes^2 T, T))$, but it is in fact in  the direct factor $H^1(M, \Hom(   \Sym^2 T, T))$: the piece in $H^1(M, \Hom(\exterior{2} T, T))$
vanishes because the corresponding $\Hom(\exterior{2} T, T)$-torsor is canonically trivialized by sending a connection to its torsion.)

Let $D^{ i}$ be the sheaf of differential operators of order ${\leq i}$ on $M$. 
We have a short exact sequence of locally free sheaves on $M$:
\begin{equation}\label{D2bos}
0 \to T \to D^{ 2} / D^{ 0} \to \Sym^2T \to 0,
\end{equation}
where we identify $ D^{ 1} / D^{ 0} $ with $T$ and $ D^{ 2} / D^{1} $ with $\Sym^2T$. 
We let $[D^2]$ be the extension class of this extension.

Finally, let 
\begin{equation} \label{xwae}
M(A_i) := \Maps(\Spec(A_i),M)
\end{equation}
be the space of $A_i$-valued points of $M$, where $A_i$ is
the (commutative) Artinian ring $\CC[\epsilon_1, \dots, \epsilon_i]/(\epsilon_i^2,\quad i=1,\dots,i)$. So $M(A_0) = M,  M(A_1)= T_M$, and $M(A_2)$ will be our bosonic analogue of $S(B_2)$. It is an affine $T$ bundle over the total space of $T \oplus T$. Proposition \ref{interpretations of obstruction} has a straightforward analogue:

\begin{proposition}\label{interpretations of Atiyah class}
 The following three classes in $H^1(M, \Hom(\Sym^2 T, T))$ are equal, for any complex manifold $M$:
\begin{enumerate}
\item The Atiyah class $\alpha_{T_M}$;
\item The extension class $[D^2]$ of \eqref{D2bos}, given by the sheaf of differential operators $ D^{ 2} / D^{ 0} $ on $M$;
\item The class of the deformation space $M(A_2)$ as an affine $T_+$-bundle over the total space $Y$ of $T_{-} \oplus  T_{-}$ over $M$. 
\end{enumerate}
\end{proposition}

More generally, to a (holomorphic) vector bundle $V$ on a complex manifold $M$ is associated its Atiyah class:
\[
\alpha_V= \alpha_{M,V} \in H^1(M, T^{\vee} \otimes \End(V)  )
\]
which is the class of the sheaf $\Conn(V)$ whose sections on an open $U \subset M$ consist of all (holomorphic) connections on $V_U$. This is the sheaf of sections of an affine bundle modeled on $ T^{\vee} \otimes \End(V) $. In the special case that $V=T$, the notion of torsion free connections reduces the underlying vector bundle from 
$ T^{\vee} \otimes \End(T)  = T^{\vee} \otimes T^{\vee} \otimes T$  to 
$\Hom(   \Sym^2 T, T) =   \Sym^2 T^{\vee} \otimes T.$ Both the general case and the special case $V=T$ are studied in \cite{At} and reviewed succinctly in \cite{K}, where several additional interpretations of the Atiyah classes are given.
One of these involves the  sheaf $ D^{ 1} \otimes_{\OO} V^{\vee}$ of  first order differential operators 
from $ V$ to $\OO$. This fits into the short exact sequence:
\begin{equation}\label{D1bos}
0 \to V^{\vee} \to D^{ 1} \otimes_{\OO} V^{\vee} \to T  \otimes V^{\vee} \to 0.
\end{equation}
In case $V=T$, the symmetric part of this sequence matches sequence \eqref{D2bos}. 
A partial generalization of 
Proposition \ref{interpretations of Atiyah class} is:

\begin{proposition}\label{interpretations of Atiyah V class}
 The following classes in $H^1(M, T^{\vee} \otimes \End(V))$ are equal, for any complex manifold $M$ and holomorphic bundle $V$ on it:
\begin{enumerate}
\item The Atiyah class $\alpha_{M,V}$;
\item The extension class $[D^1 \otimes_{\OO} V^{\vee}]$ of \eqref{D1bos}, given by the sheaf of first order differential operators 
from $ V$ to $\OO$.
\end{enumerate}
\end{proposition}
\begin{proof}  See \cite{K}, (1.1.4).
\end{proof} 

\subsection{Super Atiyah class}\label{cghjjlkh}

So far, the analogy between our bosonic Proposition \ref{interpretations of Atiyah class} 
and the supermanifold version Proposition \ref{interpretations of obstruction} 
may seem somewhat  mysterious. 
We wil now see that the Atiyah class of $M$ and the obstruction class of the supermanifold $S$
are two (of the three) components of a larger object, the super Atiyah class of $S$ .

The definition of the Atiyah class, and the analogue of Proposition \ref{interpretations of Atiyah class},
extend immediately to the super world.
Let $S$ be a supermanifold with tangent bundle $T=T_S$. 
The (super) Atiyah class $\alpha_V$ of a  (super) vector bundle $V$ on $S$ is defined exactly as before: 
it is the class, in $H^1(S, \ T^{\vee} \otimes \End(V) \  ) $, of the sheaf $\Conn(V)$ 
whose sections on an open $U \subset M$ consist of all connections on $V_U$. 
As in the bosonic case, this is the sheaf of sections of an affine bundle modeled on $ T^{\vee} \otimes \End(V) $.
And again, in the special case that $V=T_S$, 
the notion of torsion free connections reduces the  vector bundle underlying  $\Conn(T)$  from 
$ T^{\vee} \otimes \End(T)  = T^{\vee} \otimes T^{\vee} \otimes T$  to 
$\Hom(   \Sym^2 T, T) =   \Sym^2 T^{\vee} \otimes T.$ 
(Here $T$ is a graded vector bundle and the symmetric product  is in the graded sense.) 

Similarly, we define the sheaf $\D^{ i}$ of differential operators of order ${\leq i}$ on $S$,
and  the extension class $[\D^2]$ of the  extension corresponding to \eqref{D2bos}.

Finally, let $A$ be a graded Artinian algebra. 
We want to define the supermanifold $\wt{S}(A)$ of $A$-valued points of $S$.
It should parametrize superfamilies over $A$ of points of $S$. 
The first guess might  be the space $S(A) := \Maps(\Spec(A),S)$ of $A$-valued points of $S$.
This is wrong: it is a set, and has a natural structure of a manifold, 
but we want a supermanifold that keeps track of both even and odd aspects of $S$. 
The problem is that maps of supermanifolds are by definition always even,
so $S(A) = \Maps(\Spec(A),S)$ gives only the reduced space of the desired $\wt{S}(A)$. 
Instead, we need to consider what amounts to 
internal-$\Hom$ (or internal-$\Maps$) in the category of supermanifolds.
Namely, $\wt{S}(A)$ is defined to be the supermanifold $\underline{\Maps}(\Spec(A),S)$ representing 
the {\em functor} of  $A$-valued points of $S$. 
The defining property is that for every supermanifold $X$ there should be a natural identification:
\begin{equation} \label{xsdfg}
\Maps(X \times \Spec A, S) = \Maps(X,  \underline{\Maps}(\Spec(A),S)).
\end{equation}
This property is strong enough that any two candidates for $\underline{\Maps}(\Spec(A),S)$ are naturally identified. 
In particular, it gives a natural way to build $\wt{S}(A)=\underline{\Maps}(\Spec(A),S)$ by glueing $\wt{S_1}(A)$ and $\wt{S_2}(A)$ 
whenever $S$ is the union of open subsets $S_1$ and $S_2$,
by identifying the restrictions to $\S_1 \cap S_2$.
So it suffices to construct $\wt{S}(A)$ when $S$ is affine (or Stein). 
When our $S$ happens to be a super vector space $W$, we can take 
$\underline{\Maps}(\Spec(A),W)$ to be $A \otimes W$,
the tensor product taken as Z/2-graded vector spaces. 
In general, the affine $S$ is defined in some super vector space $W$ by an ideal $I$. 
We then take $\underline{\Maps}(\Spec(A),S)$ to be the subspace of 
$\underline{\Maps}(\Spec(A),W)$
defined by the ideal:
$\wt{I} := (f \otimes r | f \in A^*, r \in I),$
with $A^*$ the super vector space of linear functions on $A$.

We are interested in $\wt{S}(A_i)$, the supermanifold of $A_i$-valued points of $S$, 
where as before  $A_i$ is the (purely even) Artinian ring 
$\CC[\epsilon_1, \dots, \epsilon_i]/(\epsilon_i^2,\quad i=1,\dots,i)$. 
So $\wt{S}(A_0) = S$,  and we see from the construction above that 
$\wt{S}(A_1)= T_S$, the total (super)space of the tangent bundle of $S$. 
Finally, $\wt{S}(A_2)$ will be our super version of $M(A_2)$. 
It is an affine $T_S$ bundle over the total (super)space of $T_S \oplus T_S$. 
The extension of Proposition \ref{interpretations of Atiyah class} to supermanifolds is trivial:

\begin{proposition}\label{interpretations of super Atiyah class}
 The following three classes in $H^1(S, \Hom(\Sym^2 T_S, T_S))$ are equal, for any complex supermanifold $S$:
\begin{enumerate}
\item The super Atiyah class $\alpha_{T_S}$;
\item The extension class  $[\D^2]$ of the sheaf of differential operators $ \D^{ 2} / \D^{ 0} $ on $S$;
\item The class of the deformation space $\wt{S}(A_2)$ as an affine $T_S$-bundle over the total space $Y$ of $T_{S} \oplus  T_{S}$ over $S$. 
\end{enumerate}
\end{proposition}

More generally, for any vector bundle $V$ on our supermanifold $S$ we have the obvious analogue of 
Proposition  \ref{interpretations of Atiyah V class}. When $V=T$, the extension classes in 
Proposition  \ref{interpretations of Atiyah V class}(2) and 
Proposition \ref{interpretations of super Atiyah class}(2) agree.

The interesting new feature arises when we restrict this super Atiyah class back to the reduced space $M$. 
The decomposition of the tangent bundle:
\[
(T_S)_{|M} = T_+ \oplus T_-
\]
causes the cohomolohy group to decompose as the sum of three pieces:
\begin{align}
&H^1(M, \Hom(\Sym^2 T_S, T_S)) =   \label{pieces}
\\
\notag H^1(M, \Hom( \Sym^2 T_+,T_+) ) & \oplus 
H^1(M,  \Hom(\exterior{2}T_{-},T_+ ) ) \oplus  
H^1(M,   \Hom(T_+ \otimes T_-, T_- )  ). 
\end{align}
The point is that we can identify the three components of the super Atiyah class:

\begin{theorem}\label{components of super Atiyah class}
 The three components of the super Atiyah class $\alpha_{T_S}$ of a supermanifold
$S=(M,\OO_S)$ under the decomposition \eqref{pieces} are:
\begin{enumerate}
\item The Atiyah class $\alpha_{T_M}$ of the reduced manifold $M$;
\item The obstruction class  $\omega(S)$ to splitting $S$;
\item The Atiyah class $\alpha_{T_-}$ of the vector bundle $V=T_-$ on $M$.
\end{enumerate}
\end{theorem}
\begin{proof}
We could use part (3) of Proposition \ref{interpretations of super Atiyah class}
to replace $\alpha_{T_S}$ by the map space $\wt{S}(A_2)$.  
We recover $M(A_2)$ as the reduced space of $\wt{S}(A_2)$:
\begin{align*}
\wt{S}(A_2)_{red} &= \Maps(point, \wt{S}(A_2))  \\ 
                           &= \Maps(\Spec(A_2),S) \qquad  &&{\text {(this is a  special case of \eqref{xsdfg})}} \\
                           &= \Maps(\Spec(A_2),M)  \qquad && {\text {(since $A_2$ is even and $M=S_{red}$ )}} \\
                           &= M(A_2)  \qquad  &&{\text {(by the definition \eqref{xwae}).}}
\end{align*}
We recall that  $\wt{S}(A_2)$ is an affine $T_S$ bundle over the total space of $T_S \oplus T_S$. 
When this is restricted to $M$, $T_S$ decomposes as $T_+ \oplus T_-$. 
The affine bundle splits into four pieces (the even ones, i.e. those involving an even number of $T_-$'s).
The two mixed pieces among these coincide, and can be identified with $\alpha_{T_-}$. 
The two other  pieces can  be identified with the map spaces 
$M(A_2) = S(A_2) = \wt{S}(A_2)_{red}$ and $S(B_2)$, respectively. 
The details become somewhat complicated, partly due to the interpretation of the mixed piece,
so we follow a different route.

Instead, we use part (2) of Proposition \ref{interpretations of super Atiyah class}
to replace $\alpha_{T_M}$ by the  extension class  $[\D^2]$ of 
the sheaf of differential operators $ \D^{ 2} / \D^{ 0} $ on $S$ occurring in \eqref{D2bos}.
In this language, the analogue of the decomposition \eqref{pieces} involves the three invariant submodules
$ \D^{ 2}_+, \D^{ 2}_-, \D^{ 2}_{\pm}$ of $\D^{ 2}$:

\begin{align*}
 \D^{ 2}_+ \cong \OO_M< \partial^2/ (\partial x)^2, \partial/\partial x >\\
 \D^{ 2}_- \cong \OO_M< \partial^2/ (\partial \theta)^2, \partial/\partial x > \\
 \D^{ 2}_{\pm} \cong \OO_M< \partial^2/ \partial x \partial \theta, \partial/\partial \theta >.
\end{align*}
Their invariance on $M$, i.e. modulo the nilpotent ideal $J$, under a coordinate change:
\begin{align*}
\wt{x} = \wt{x}(x,\theta) \\
\wt{\theta}  = \wt{\theta}(x,\theta) 
\end{align*}
follows immediately from the fact that the off-diagonal coefficients  in the transformation laws:
\begin{align*}
\partial/\partial x_i = {\partial \wt{x}_{\wt{i}} /\partial x_i } \ {  \partial  /\partial \wt{x}_{\wt{i}}  } 
+ {\partial \wt{\theta}_{\wt{j}} /\partial x_i } \ {  \partial  /\partial \wt{\theta}_{\wt{j}}  } \\
\partial/\partial \theta_j = {\partial \wt{x}_{\wt{i}} /\partial \theta_j } \ {  \partial  /\partial \wt{x}_{\wt{i}}  } 
+  {\partial \wt{\theta}_{\wt{j}}/\partial \theta_j } \ {  \partial  /\partial \wt{\theta}_{\wt{j}} } 
\end{align*}
are odd, hence in $J$. So for example to see the invariance of $\D^{2}_- $, we note that:
\begin{equation} \label{yfouyfo}  
\begin{aligned}
&\frac {\partial^2} {\partial \theta_1 \partial \theta_2} =  \\
& (\frac { \partial \wt{x}_{\wt{i}}} {\partial \theta_1 } \ \frac { \partial} {\partial \wt{x}_{\wt{i}}  } 
+  \frac{\partial \wt{\theta}_{\wt{j}}} {\partial \theta_1 } \ \frac{  \partial} {\partial \wt{\theta}_{\wt{j}} } )
 (\frac{ \partial \wt{x}_{\wt{i'}}} {\partial \theta_2 }  \frac{  \partial} {\partial \wt{x}_{\wt{i'}}  } 
+ \frac {\partial \wt{\theta}_{\wt{j'}}}  {\partial \theta_2 } \ \frac{  \partial} {  \partial \wt{\theta}_{\wt{j'}} }) = \\
&\frac {\partial \wt{\theta}_{\wt{j}}} {\partial \theta_1 } \  \frac{\partial \wt{\theta}_{\wt{j'}}} {\partial \theta_2 } 
\ \frac {\partial^2} {  \partial \wt{\theta}_{\wt{j}}   \partial \wt{\theta}_{\wt{j'}} } +
\frac{\partial \wt{\theta}_{\wt{j}}} {\partial \theta_1 }   \frac{ \partial^2  \wt{x}_{\wt{i'}}} { \partial \wt{\theta}_{\wt{j}} \partial \theta_2}  \frac{\partial } { \partial  \wt{x}_{\wt{i'}}}  + 
\text{6 more terms in }  J. 
\end{aligned}
\end{equation}
It is now clear that the extension in the super version of \eqref{D2bos}, which represents the super Atiyah class, splits as a direct sum of three extensions, each involving one of 
the three invariant subrings 
$ \D^{ 2}_+, \D^{ 2}_-, \D^{ 2}_{\pm}$ of $\D^{ 2}$.
The identification of the corresponding extensions with the three summands in \eqref{pieces}
now follows from Proposition \ref{interpretations of  Atiyah class}, Proposition \ref{interpretations of  obstruction}, 
and Proposition \ref{interpretations of  Atiyah V class} for $T_-$, respectively.

\end{proof}

\section{Obstruction and Atiyah classes for moduli spaces} \label{3}

In this section we specialize the previous results to obtain  concrete cohomological descriptions 
of the obstruction and Atiyah classes of moduli spaces of (super) Riemann surfaces. 
We start in \ref{3,3,1}  by describing a useful class of short exact sequences on the product $C \times C$. 
In the next two sections we use it to describe the obstruction and Atiyah classes of moduli spaces. 
The proof for the Atiyah class is given in detail in section \ref{oyf}.
The proof for the obstruction class is similar though somewhat easier; it is outlined in section \ref{sfg}.
Ideally, there should be an analogous expression for the super Atiyah class as well, 
from which the two previous results should follow. 
For now, this remains open.

Our aproach is as follows. As we saw in section \ref{2}, the first obstruction to the splitting  of  a supermanifold $S=(M,\OO_S)$ is a class
\begin{equation}
\label{omega} 
\omega \in H^1(M, \Hom ( \exterior{2}T_-,T_+)).
\end{equation}
We may interpret any such class as the extension class of a short exact sequence of vector bundles on $M$:
\begin{equation} \label{E}
0 \to T_+ \to E \to \exterior{2}T_- \to 0
\end{equation}
or equivalently of the dual sequence:
\begin{equation}\label{E*}
0 \to \exterior{2}T_{-}^* \to E^* \to T_{+}^* \to  0.
\end{equation}
In  section  \ref{sxdfgh} we give an algebro-geometric interpretation 
of the bundles $E, E^*$ and the sequences \eqref{E},\eqref{E*} 
in the case that $S =\SM_g, M=\sM_g.$
A similar general description as extension class applies to the Atiyah class of a manifold.
We give the corresponding algebro-geometric interpretation 
for the Atiyah class of the moduli space $\M_g$ of Riemann surfaces
in section  \ref{Bosonic detour}.

\subsection{Short exact sequences on $C \times C$}\label{3,3,1} 
\mbox {   }

 We fix a spin curve $(C,T_C^{1/2}) \in \sM_g$. Given integers $a,b,c$, we consider the line bundle $\OO(a,b,c)$ on the surface $C \times C$:
\begin{equation}
\label{abc} 
\OO(a,b,c) := {{p_1}^* {K_C}^{\otimes a/2} } \otimes  { {p_2}^* {K_C}^{\otimes b/2} }    
\otimes   {\OO_{C \times C} (c \Delta)},
\end{equation}
where $p_1, p_2 : C \times C  \to C$ are the projections, and $\Delta$ is the diagonal. Restriction to the diagonal gives our basic short exact sequence:
\begin{equation} \label{SES}
0 \to \OO(a,b,c-1) \to \OO(a,b,c) \stackrel{\Res} {\to}
{( K_C)} ^ {\otimes (a+b-2c)/2} \to 0.
\end{equation}
This is a sequence of coherent sheaves on $C \times C$. The first two sheaves are line bundles. The third is a line bundle on the diagonal, interpreted as a sheaf on  $C \times C$ supported on  the diagonal. 
The map denoted $\Res$ can be interpreted either as a restriction to the diagonal, or as a residue. 

When $a=b$, we are going to decompose sequence \eqref{SES} under the action of an involution.
Let $z$ be a local coordinate on $C$, 
and let $x:=p_1^*(z), y:=p_2^*(z)$ be the corresponding local coordinates on  $C \times C$. 
We will need to identify 
${p_1}^* {K_C}^{\otimes a/2}  \otimes   {p_2}^* {K_C}^{\otimes b/2} $
with
$  {p_2}^* {K_C}^{\otimes b/2}    \otimes   {p_1}^* {K_C}^{\otimes a/2} $. 
This involves a more or less arbitrary choice of sign.
We will stick with the usual rule of signs: 
\begin{equation}\label{ffdncvc}
{dx^{\otimes \frac{a}{2}} dy ^{\otimes \frac{b}{2}}} \mapsto
(-1)^{ab}{dy ^{\otimes \frac{b}{2}} dx^{\otimes \frac{a}{2}} }.
\end{equation}
(Some formulas below would actually be simpler if we omit the 
$(-1)^{ab}$ sign. We keep it to emphasize the fermionic nature of half differentials.)
Having made this choice, the natural involution $(x,y) \to (y,x)$ of $C \times C$ 
then sends $ \OO(a,b,c)  $ to  $ \OO(b,a,c)  $, so it 
acts on $\OO_{C \times C}(a,a,c)$ and on its cohomology; 
we indicate the eigenspaces with $\pm$ superscripts. 
Locally, a section is of the form 
\begin{equation} \label{explicit}
\varphi=f(x,y)\frac{dx^{\otimes \frac{a}{2}} dy ^{\otimes \frac{a}{2}}} {(x-y)^{c} },
\end{equation}
with holomorphic $f$. 
It is even under the involution if the parity of the function $f$ is the same as that of the integer $a-c$,
while if these parities are opposite the section is odd under the involution.
The involution acts on sequence \eqref{SES}$_{aac}$, which breaks as a sum of its even and odd pieces.
Note that the third term $K_C^{a-c} $ has pure sign $(-1)^{a-c}$, since
if $f$ in \eqref{explicit} is antisymmetric then $\Res (\varphi) =0$. So sequence 
\eqref{SES}$_{aac}$ breaks as a sum of its $(-1)^{a-c}$ subsequence:
\begin{equation}\label{abc+}
0 \to \OO(a,a,c-1)^{(-)^{a-c}} \to \OO(a,a,c)^{(-)^{a-c}}   \stackrel{\Res} {\to}   K_C^{a-c} \to 0
\end{equation}
and its {\em trivial}  $(-1)^{a-c-1}$ subsequence:
\begin{equation}\label{abc-}
0 \to \OO(a,a,c-1)^{(-)^{a-c-1}} \to \OO(a,a,c)^{(-)^{a-c-1}}   \stackrel{\Res} {\to}  0  {\to}  0.
\end{equation}

\subsection{A cohomological interpretation of $\mathbf{\omega_{\SM_g  }}$}\label{sxdfgh}
\mbox {   }

Consider the sequence  {\eqref{SES}}$_{3,3,1}$ :
\begin{equation}\label{331}
0 \to \OO(3,3,0) \to \OO(3,3,1) \to  K_C{^2} \to 0.
\end{equation}
The three sheaves occurring here have no higher cohomology,  so the long exact sequence is:
\begin{equation} \label{LES331}
0 \to H^0(  C{\times}C,  \OO(3,3,0) )\to H^0(  C{\times}C,  \OO(3,3,1) )  \to H^0(C, {K_C{^2})}  \to 0,           
\end{equation}
or equivalently:
\[
0 \to H^0(C,  K_C^{3/2}) ^{\otimes 2}     \to H^0(  C{\times}C,  \OO(3,3,1) )  \to H^0(C, {K_C{^2})})  \to 0.
\]
Under the action described in section \ref{3,3,1} of the  involution 
\[
(x,y) \to (y,x)
\]
of  $C \times C$, the sequence splits into a trivial odd part {\eqref{LES331}}$^-$:
\begin{align*}
0 \to & \Sym^2(H^0(C,  K_C^{3/2}))        &   \to  & \Sym^2(H^0(C,  K_C^{3/2}))    & \to &   &         & 0 & \to 0 \\
\intertext{and a non-trivial even part {\eqref{LES331}}$^+$  :} 
0 \to & \exterior{2}(H^0(C,  K_C^{3/2}))  &   \to & H^0(  C{\times}C,  \OO(3,3,1) )^+ & \stackrel{\Res}{\to}&  &H^0(&C,{K_C{^2})})  & \to 0.
\end{align*}
(The somewhat awkward interpretation of the symmetric (respectively antisymmetric) square 
as the odd (respectively even) subspaces under the involution results from our convention \eqref{ffdncvc}. If we had dropped the $(-1)^{ab}$ sign there, the identification here would have flipped.)

The Serre dual (i.e. the dual tensored with the canonical line bundle) of the inclusion
\[
\OO(3,3,0)  \to  \OO(3,3,1)
\]
is the inclusion
\[
\OO(-1,-1,-1)  \to  \OO(-1,-1,0)
\]
which is part of the sequence  {\eqref{SES}}$_{-1,-1,0}$ :
\begin{equation}\label{-1-10}
0 \to \OO(-1,-1,-1) \to \OO(-1,-1,0) \to  K_C{^{-1}} \to 0.
\end{equation}
The long exact sequence is now:
\begin{equation} \label{LES-1-10}
0  \to H^1(C, T_C)  \to H^2(  C{\times}C,  \OO(-1,-1,-1) )\to H^2(  C{\times}C,  \OO(-1,-1,0))    \to 0,
\end{equation}
which is the dual of \eqref{LES331}. This can also be written:
\[
0  \to H^1(C, T_C)  \to H^2(  C{\times}C,  \OO(-1,-1,-1) )\to (H^1(  C,  T_C^{1/2} ))^{\otimes 2}   \to 0.
\]
Under the involution  this again splits into a trivial odd part {\eqref{LES-1-10}}$^-$:
\begin{align*}
0 \to  && 0                  & \to  & \Sym^2((H^1(  C,  T_C^{1/2} ))            & \to & \Sym^2((H^1(  C,  T_C^{1/2} ))  & \to &0 \\
\intertext{and a non-trivial even part {\eqref{LES-1-10}}$^+$  :} 
0  \to && H^1(C, T_C)  &\to & (H^2(  C{\times}C,  \OO(-1,-1,-1) ))^- & \to &\exterior{2}(H^1(  C,  T_C^{1/2} )) &   \to &0,
\end{align*}

\begin{proposition} \label{CxC} 
(1) The fiber at $(C,T_C^{1/2}) \in \sM_g$ of the extension \eqref{E*} expressing the first obstruction $\omega$ to the splitting  of $\SM_g$ is canonically identified with the even part {\eqref{LES331}}$^+$ of the long exact sequence  {\eqref{LES331}} of the restriction-to-the-diagonal short exact sequence   \eqref{331} =  {\eqref{SES}}$_{3,3,1}$.
\newline (2) Dually, the fiber at $(C,T_C^{1/2}) \in \sM_g$ of the extension \eqref{E} given by the first obstruction $\omega$ to the splitting  of $\SM_g$ is canonically identified with the even part {\eqref{LES-1-10}}$^+$ of the long exact sequence  {\eqref{LES-1-10}} of the restriction-to-the-diagonal short exact sequence  \eqref{-1-10} =  {\eqref{SES}}$_{-1,-1,0}$.
\end{proposition}

Our claim that these identifications are ``canonical'' means that they remain valid in families. 
In particular, the sequence \eqref{D2} on $\sM_g$, expressing the first obstruction class \eqref{omega} to the projectedness and splitting of $\SM_g$, 
can be constructed as follows. We start with the universal families 
$\pi:  \sM_{g,1} \to \sM_g$
and \break
$\pi\pi: \sM_{g,1} \times_{\sM_{g}}   \sM_{g,1} \to \sM_g$ 
over $\sM_{g}$, whose typical fibers are $C$ and  $C \times C$, respectively. 
Let   $T_{\pi}, K_{\pi}$ denote the relative tangent and cotangent bundles,
and  let $p_1,p_2$ be the two projections of 
$\sM_{g,1} \times_{\sM_{g}}   \sM_{g,1}$ to $\sM_{g,1}$.
On the  total space 
$\sM_{g,1} \times_{\sM_{g}}   \sM_{g,1} $ 
there are short exact sequences analogous to   \eqref{331} =  {\eqref{SES}}$_{3,3,1}$ 
and \eqref{-1-10} =  {\eqref{SES}}$_{-1,-1,0}$: 
\begin{equation}\label{331universal}
0 \to \wt{\OO}(3,3,0) \to \wt{\OO}(3,3,1) \to K_ {\pi}^2       \to 0
\end{equation}
and
\begin{equation}\label{*331universal}
0 \to \wt{\OO}(-1,-1,-1) \to \wt{\OO}(-1,-1,0) \to T_ {\pi}       \to 0,
\end{equation}
where we set:
\[
\wt{\OO}_{\sM_{g,1} \times_{\sM_{g}}   \sM_{g,1}}(a,b,c) := p_1^*(K_{\pi})^{a/2} \otimes p_2^*(K_{\pi})^{b/2} \otimes \OO_{\sM_{g,1} \times_{\sM_{g}}   \sM_{g,1}}( c \Delta),
\]
with $\Delta$ now denoting the diagonal copy of
$\sM_{g,1}$
in 
$\sM_{g,1} \times_{\sM_{g}}   \sM_{g,1}$.

On each  $C \times C$ fiber, these restrict to the previous sequences \eqref{331},\eqref{-1-10}.
The direct images on $\sM_g$ give the dual sequences:
\begin{equation}\label{331onMg}
0 \to \pi\pi_*\wt{\OO}(3,3,0) \to \pi\pi_*\wt{\OO}(3,3,1) \to T^*_{\sM_g}       \to 0
\end{equation}
and
\begin{equation}\label{*331onMg}
0  \to T_{\sM_g}  \to R^2\pi\pi_*\wt{\OO}(-1,-1,-1) \to R^2\pi\pi_*\wt{\OO}(-1,-1,0) \to  0,
\end{equation}
so we conclude:
\begin{theorem}\label{obstruction for SM_g}
The obstruction class of the moduli space $\sM_g$ is given by the even (=antisymmetric!) part of the extension class of sequence \eqref{331onMg}, or, dually, \eqref{*331onMg}.
\end{theorem}

\subsection{Bosonic version: an interpretation of $\alpha_{\M_g}$} \label{Bosonic detour}

\mbox {   }

Before we prove Proposition \ref{CxC}, we pause to discuss the analogous result in the bosonic world. There is a large literature devoted to describing higher order neighborhoods of points in various moduli spaces and the functions and differential operators on them   \cite{BS}-\cite{BB}.
The result we have in mind describes the Atiyah class $\alpha(\M_g)$, or any of the equivalent objects appearing in Proposition \ref{interpretations of Atiyah class}, for the moduli space $\M_g$ of complex structures on a Riemann surface $C$. We do this in terms of the sheaves $\OO_{C \times C}(a,b,c)$ occurring in the extensions \eqref{SES} on $C \times C$. We were unable to find this result in the literature, so we thought it was worth discussing here.{\footnote{ A related result appears in \cite{BS}, where  the authors calculate the Atiyah class of the {\em canonical bundle} of the moduli space $\M_g$ of complex structures on a Riemann surface $C$.
This is determined by the Atiyah class $\alpha(\M_g)$ of the full {\em tangent bundle} of $\M_g$, which is what we find here, but it does not fully determine $\alpha(\M_g)$.
Beilinson has explained  that nevertheless, our result about  $\alpha(\M_g)$ can  be deduced from the details of the proof of the main theorem of \cite{BS}. In the proof, in section 2.3.3 of \cite{BS}, Beilinson and Schechtman construct a canonical isomorphism of certain extensions, and this canonical isomorphism can be used to derive our description of $\alpha(\M_g)$.  }}

 Our usual notation $\OO_{C \times C}(a,b,c)$ for the line bundle
\[
\OO_{C \times C}(a,b,c) := p_1^*(K_C)^{a/2} \otimes p_2^*(K_C)^{b/2} \otimes \OO_{C \times C}( c \Delta),
\]
with $\Delta$ the diagonal, seems inappropriate for the bosonic case,  where there are no natural spin structures. So we use the alternative notation
\[
\OO_{C \times C}(a*,b*,c) := \OO_{C \times C}(2a,2b,c) = p_1^*(K_C)^{a} \otimes p_2^*(K_C)^{b} \otimes \OO_{C \times C}( c \Delta). 
\]
These fit into the restriction-to-$\Delta$ (or: residue) short exact sequence:
\begin{equation}\label{a*b*c}
0 \to \OO(a*,b*,c-1) \to \OO(a*,b*,c)  \stackrel{\Res} {\to}   K_C^{a+b-c} \to 0.
\end{equation}
As previously, let $z$ be a local coordinate on $C$, 
and let $x:=p_1^*(z), y:=p_2^*(z)$ be corresponding local coordinates on  $C \times C$. 
The involution $(x,y) \to (y,x)$ of $C \times C$ 
now sends $ \OO(a*,b*,c)  $ to  $ \OO(b*,a*,c)  $ naturally, 
without the ambiguity we encountered in equation \eqref{ffdncvc}.
In particular, it acts on $\OO_{C \times C}(a*,a*,c)$ and on its cohomology; 
we indicate the eigenspaces with $\pm$ superscripts. 
Locally, a section is of the form: 
\begin{equation}\label{sde}
\varphi = f(x,y)\frac{dx^{\otimes a} dy ^{\otimes a} }{(x-y)^{c} },
\end{equation}
 with holomorphic $f$. 
It is even under the involution if the parity of the function $f$ is the same as that of the integer $c$,
while if these parities are opposite the section is odd under the involution.
The involution acts on sequence \eqref{a*b*c}, which breaks as a sum of its even and odd pieces.
Now  the third term $K_C^{2a-c} $ has pure sign $(-1)^{c}$. So sequence 
\eqref{a*b*c} breaks as a sum of its $(-1)^{c}$ subsequence:
\begin{equation}\label{a*b*c+}
0 \to \OO(a*,a*,c-1)^{(-)^c} \to \OO(a*,a*,c)^{(-)^c}   \stackrel{\Res} {\to}   K_C^{2a-c} \to 0
\end{equation}
and its trivial  $(-1)^{c-1}$ subsequence:
\begin{equation}\label{a*b*c-}
0 \to \OO(a*,a*,c-1)^{(-)^{c-1}} \to \OO(a*,a*,c)^{(-)^{c-1}}   \stackrel{\Res} {\to}  0  {\to}  0.
\end{equation}

The case relevant to us is $a=b=c=2$:
\begin{equation}\label{442}
0 \to \OO(2*,2*,1) \to \OO(2*,2*,2)  \stackrel{\Res} {\to}   K_C^2 \to 0.
\end{equation}
For $g>1$, these sheaves have no higher cohomology, so the long exact sequence is
\[
0 \to H^0(\OO(2*,2*,1)) \to H^0(\OO(2*,2*,2)) \to   H^0(K_C^2) \to 0,
\]
which decomposes into an uninteresting odd part: 
\[
H^0(\OO(2*,2*,1)) ^-\cong H^0(\OO(2*,2*,2))^-,
\]
and an even part:
\begin{equation}\label{H442}
0 \to  \Sym^2 H^0(K^2) \to H^0(\OO(2*,2*,2))^+ \to H^0(K_C^2) \to 0.
\end{equation}
Here the identification $ \Sym^2 H^0(K^2) \cong H^0(\OO(2*,2*,1))^+ $ comes from 
sequence \eqref{a*b*c-} for $a=b=2, c=1$, i.e. it is
the even part of the long exact sequence
\[
0 \to H^0(\OO(2*,2*,0)) \to H^0(\OO(2*,2*,1)) \to H^0(K_C^3) \to 0,
\]
coming from the short exact
\[
0 \to \OO(2*,2*,0) \to \OO(2*,2*,1) \to K_C^3 \to 0,
\]
where we recall that $H^0(K_C^3) $ is odd under the involution.
To dualize the sequence \eqref{442}, we take the Serre dual of the first map (between the line bundles) and find the new cokernel:
\begin{equation}\label{*442}
0 \to \OO(-1*,-1*,-2) \to \OO(-1*,-1*,-1) \to T_C \to 0.
\end{equation}
The even part of the long exact sequence is now the dual of \eqref{H442}:
\begin{equation}\label{H*442}
0 \to H^1(T_C) \to H^2(\OO(-1*,-1*,-2))^+ \to \Sym^2 H^1 (T_C) \to 0.
\end{equation}
Our claim is that  the extension \eqref{H*442} can be naturally identified with the fiber at (the isomorphism class of) $C$ of the extension \eqref{D2bos} on the moduli space:

\mbox{  }

\begin{proposition}\label{CxCbosonic} \mbox{   }
\begin{enumerate}
\item The fiber at $C \in \M_g$ of the extension \eqref{D2bos} expressing the Atiyah class $\alpha$  of the moduli space $\M_g$ is canonically identified with the even part {\eqref{H*442}} of the long exact sequence  of the restriction-to-the-diagonal short exact sequence   \eqref{*442}.
\item Dually, the fiber at $C \in \M_g$ of  the extension 
expressing the Atiyah class $\alpha$  of the moduli space $\M_g$ is canonically identified with 
the even part {\eqref{H442}} of
the long exact sequence   of the restriction-to-the-diagonal short exact sequence  \eqref{442}. 
\end{enumerate}
\end{proposition}

Again, our claim that these identifications are "canonical" means that they remain valid in families. 
In particular, the sequence \eqref{D2bos} on $\M_g$, expressing the Atiyah class of $\M_g$, 
can be constructed as follows. We start with the universal families 
$\pi:  \M_{g,1} \to \M_g$
and
$\pi\pi: \M_{g,1} \times_{\M_{g}}   \M_{g,1} \to \M_g$ 
over $\M_{g}$ whose typical fibers are $C$ and  $C \times C$, respectively. 
On the  total space 
$\M_{g,1} \times_{\M_{g}}   \M_{g,1} $ 
there are short exact sequences analogous to  \eqref{442} and \eqref{*442}:
\begin{equation}\label{442universal}
0 \to \wt{\OO}(2*,2*,1) \to \wt{\OO}(2*,2*,2) \to K_ {\pi}^2       \to 0
\end{equation}
and
\begin{equation}\label{*442universal}
0 \to \wt{\OO}(-1*,-1*,-2) \to \wt{\OO}(-1*,-1*,-1) \to T_ {\pi}       \to 0,
\end{equation}
where with the obvious modifications of the notation,
\[
\wt{\OO}_{C \times C}(a*,b*,c) := p_1^*(K_{\pi})^{a} \otimes p_2^*(K_{\pi})^{b} \otimes \OO_{\M_{g,1} \times_{\M_{g}}   \M_{g,1}}( c \Delta),
\]
and $T_{\pi}, K_{\pi}$ denote the relative tangent and cotangent bundles.
On each  $C \times C$ fiber, these restrict to the previous sequences \eqref{442},\eqref{*442}.
The direct images on $\M_g$ give:
\begin{equation}\label{442onMg}
0 \to \pi\pi_*\wt{\OO}(2*,2*,1) \to \pi\pi_*\wt{\OO}(2*,2*,2) \to T^*_{\M_g}       \to 0
\end{equation}
and
\begin{equation}\label{*442onMg}
0  \to T_{\M_g}  \to R^2\pi\pi_*\wt{\OO}(-1*,-1*,-2) \to R^2\pi\pi_*\wt{\OO}(-1*,-1*,-1) \to  0,
\end{equation}
so we conclude:
\begin{theorem}\label{Atiyah for M_g}
The Atiyah class of the moduli space $\M_g$ is given by the even (=symmetric) part of the extension class of sequence \eqref{442onMg}, or, dually, \eqref{*442onMg}.
\end{theorem}

Note that all of this is in perfect analogy with the picture for the obstruction class in the previous section. 
The only change there is that we replace the central sheaf $\OO_{C \times C}(2*,2*,2)=\OO_{C \times C}(4,4,2)$ by $\OO_{C \times C}(3,3,1)$. (And the part even under the involution is symmetric here, antisymmetric there.)

\subsection{Proof for the Atiyah class} \label{oyf}

\subsubsection{Deformations} \label{tsdrfv1}

Let $C$ be a Riemann surface, or more generally a complex manifold. 
The complex structure of  $C$  is given in terms of the almost complex structure operator 
\[
J: T_{\RR}C  \to T_{\RR}C , \quad J^2=-1
\]  
on the real tangent bundle, or equivalently in terms of the decomposition
\[
 T_{\RR}^*C \otimes {\CC}  =  T^{1,0}C \oplus T^{0,1}C.
\]
Equivalently, it is specified by a $C^{\infty}$ section of $ {\mathcal{A}}^{0,1}(TC)$, which depends on $\epsilon$ and vanishes for $\epsilon =0$. Indeed,
let $\pi^{1,0}, \pi^{0,1},i^{1,0}, i^{0,1}$ be the projections and inclusions of the summands.
For a  deformation $C_{\epsilon}$ depending on a  parameter $\epsilon$, we have a family of decompositions
\[
 T_{\RR}^*C \otimes {\CC}  =  T_{\epsilon}^{1,0}C \oplus T_{\epsilon}^{0,1}C.
\]
For small $\epsilon$,  the composition $a_{\epsilon}:=\pi_{\epsilon}^{0,1} \circ i^{0,1}: T^{0,1}C \to  T_{\epsilon}^{0,1}C$ is an isomorphism, so the complex structure of $C_{\epsilon}$ determines a linear map
\[
h_{\epsilon} := (a_{\epsilon})^{-1} \circ \pi_{\epsilon}^{0,1} \circ i^{1,0}: T^{1,0} \to T^{0,1}
\]
and is detremined by it. This map can be interpreted as a $(0,1)$-form valued section of $TC$, or a $C^{\infty}$ section of $ {\mathcal{A}}^{0,1}(TC)$, which depends on $\epsilon$ and vanishes for $\epsilon =0$.  
(As in the rest of this paper, we will usually abbreviate the holomorphic tangent and cotangent bundles $T_{1,0}C$, $T^{1,0}C$ to just $TC$,  $T^*C$.) 
We write the map $h_{\epsilon}$ in local holomorphic coordinate(s) $z$ on $C$ as 
$h_{\epsilon} =  h(z,\epsilon) d\bar{z}\frac{\partial}   {\partial{z}}.$
It specifies an explicit deformation of the $\bar{\partial}$ operator 
$\bar{\partial} = d\bar{z}\frac{\partial}   {\partial\bar{z}}$:
\begin{equation}\label{A2}
\bar{\partial}_{\epsilon} := 
\bar{\partial} + h_{\epsilon} =
\bar{\partial} + ( h(z,\epsilon)  )  d\bar{z} \frac{\partial}{\partial{z}}
=  d\bar{z}  ( \frac{\partial}{\partial{\bar{z}}}     + ( h(z,\epsilon)  ) \frac{\partial}{\partial{z}}),
\quad \quad h \in  {\mathcal{A}}^{0,1}(TC).
\end{equation}
One convenient interpretation of the 'small $\epsilon$' assumption above is that  $\epsilon$ and $h_{\epsilon}$ are nilpotent elements of some (local) Artinian algebra $A$.
When we work over the Artinian ring $A_1 = \CC[\epsilon]/(\epsilon^2)$, the $C^{\infty}$ function $h(z,\epsilon)$ becomes linear in $\epsilon$:
\begin{equation}\label{A1}
\bar{\partial}_{\epsilon} 
=  \bar{\partial} +  ( \epsilon h(z)  ) d\bar{z}  \frac{\partial}{\partial{z}},
\quad \quad h \in  {\mathcal{A}}^{0,1}(TC).
\end{equation}
When we work over the Artinian algebra $A_2 = \CC[\epsilon_1, \epsilon_2]/(\epsilon_1^2, \epsilon_2^2)$ the deformation becomes:
\begin{equation} \label{DDC}
\bar{\partial}_{\epsilon}  =\bar{\partial} + (\epsilon_1 h^1(z)  + \epsilon_2 h^2(z)  + \epsilon_1 \epsilon_2 \h(z)    )d\bar{z} \frac{\partial}{\partial{z}}, \quad \quad h^1,h^2,h^{12} \in  {\mathcal{A}}^{0,1}(TC).
\end{equation}
Such a deformation is trivial if it comes from the gauge action of a global vector field $w \in {\mathcal{A}}^{0,0}(TC),$ which acts as:
\[
\bar{\partial}_{\epsilon}  \to e^{-w} \bar{\partial}_{\epsilon}   e^w.
\]
Over the Artinian algebra $A_2 = \CC[\epsilon_1, \epsilon_2]/(\epsilon_1^2, \epsilon_2^2)$, $w$ becomes:
\begin{equation} \label{gfuctv}
w= \epsilon_1 {\alpha}^1 + \epsilon_2  {\alpha}^2  + \epsilon_1 \epsilon_2  {\alpha}^{12} =
(\epsilon_1 {\alpha}^1(z)  + \epsilon_2  {\alpha}^2(z)  + \epsilon_1 \epsilon_2  {\alpha}^{12}(z)    ) \frac{\partial}{\partial{z}},
\end{equation}
so
\[
e^w = 1 + \epsilon_1 {\alpha}^1 + \epsilon_2  {\alpha}^2  + 
\epsilon_1 \epsilon_2 ( {\alpha}^{12} + \frac{1}{2} ({\alpha}^1 {\alpha}^2  +  {\alpha}^2 {\alpha}^1))
\]
and the action on $\bar{\partial}_{\epsilon}$ is given by the somewhat complicated transformation formulas:
\begin{equation}\label{transfo}
  \begin{aligned}
h^i & \mapsto h^i  + \bar{\partial} \alpha^i, \quad i=1,2 \\
h^{12} & \mapsto h^{12} + \bar{\partial}    {\alpha}^{12} + [h^1,\alpha^2] +  [h^2,\alpha^1]  - \frac{1}{2} ( [\alpha^1 , \bar{\partial}  {\alpha}^{2}] +  [\alpha^2,  \bar{\partial}  {\alpha}^{1} ] ),
  \end{aligned}
\end{equation}
as can be seen from the calculation:
\[
\begin{aligned}
\bar{\partial}_{\epsilon}  & \mapsto && e^{-w} \bar{\partial}_{\epsilon}   e^w  \\
&= && (1 - \epsilon_1 {\alpha}^1 - \epsilon_2  {\alpha}^2  - 
\epsilon_1 \epsilon_2 ( {\alpha}^{12} - \frac{1}{2} ({\alpha}^1 {\alpha}^2  +  {\alpha}^2 {\alpha}^1))) \\
&&&(\bar{\partial} + (\epsilon_1 h^1  + \epsilon_2 h^2  + \epsilon_1 \epsilon_2 \h))\\
&&&(1 + \epsilon_1 {\alpha}^1 + \epsilon_2  {\alpha}^2  + 
\epsilon_1 \epsilon_2 ( {\alpha}^{12} + \frac{1}{2} ({\alpha}^1 {\alpha}^2  +  {\alpha}^2 {\alpha}^1))) \\
&=&& \bar{\partial} +\\
& &&\epsilon_1 ( h^1  +  \bar{\partial}(\alpha_1) )+ \\
&& &\epsilon_2  (h^2  +  \bar{\partial}(\alpha_2) )+ \\
&&& \epsilon_1 \epsilon_2 (\h + [h^1,\alpha^2] +  [h^2,\alpha^1] + \bar{\partial}    {\alpha}^{12} +  \\  
&&& \frac{1}{2}  \bar{\partial}    ({\alpha}^1 {\alpha}^2  +  {\alpha}^2 {\alpha}^1))  +  
({\alpha}^1 {\alpha}^2  +  {\alpha}^2 {\alpha}^1)  \bar{\partial}   \\
&&&-(\alpha^1   \bar{\partial}  {\alpha}^{2} + \alpha^2   \bar{\partial}  {\alpha}^{1} ) - 
 ({\alpha}^1 {\alpha}^2  +  {\alpha}^2 {\alpha}^1)  \bar{\partial}    )  \\
&=&& \bar{\partial}_{\epsilon} + \epsilon_1   \bar{\partial}(\alpha_1) + \epsilon_2  \bar{\partial}(\alpha_2) + \\
&&& \epsilon_1 \epsilon_2 ( \bar{\partial}    {\alpha}^{12} + [h^1,\alpha^2] +  [h^2,\alpha^1]  - \\
&&&\frac{1}{2} ( [\alpha^1 , \bar{\partial}  {\alpha}^{2}] +  [\alpha^2,  \bar{\partial}  {\alpha}^{1} ] )).
\end{aligned}
\]

It is of course much easier to work infinitesimally, i.e. to switch from the Lie group to its Lie algebra. In our case this simply means that we drop terms quadratic in the $\alpha$'s, so the transformation formulas \eqref{transfo} are replaced by the infinitesimal transformation formulas :
\begin{equation}\label{infinitesimaltransfo}
  \begin{aligned}
h^i & \mapsto h^i  + \bar{\partial} \alpha^i, \quad i=1,2 \\
h^{12} & \mapsto h^{12} + \bar{\partial}    {\alpha}^{12} + [h^1,\alpha^2] +  [h^2,\alpha^1].
  \end{aligned}
\end{equation}

\subsubsection{The pairing} \label{tsdrfv2}

To prove our Proposition, we exhibit a natural pairing between elements $\varphi$ of the central  cohomology group 
$ H^0(\OO(2*,2*,2))^+$ in \eqref{H442}, which is dual to the central  cohomology group in extension  \eqref{H*442}, and elements of the extension \eqref{D2bos}, which according to Proposition \ref{interpretations of Atiyah class} we describe in terms of pairs 
\[
h^1 \otimes h^2, h^{12},
\] 
where  $h^1, h^2, h^{12}  \in  {\mathcal{A}}^{0,1}(TC)$ are the parameters \eqref{DDC} on the deformation space $M(A_2)$, i.e. they are defined only modulo the infinitesimal transformation  rule \eqref{infinitesimaltransfo}. If $\varphi$ is in the kernel of the Residue map, i.e. $\varphi \in \Sym^2 H^0(K_C)$, it is clear that we can pair it with the cohomology class of $h^1 \otimes h^2$. Similarly, it is clear how to pair $\Res (\varphi)$ with the cohomology class of $h^{12}$. The subtle point is to check that the non-trivial extension   \eqref{H*442} exactly matches the non-trivial transformation rule \eqref{infinitesimaltransfo}.

\def\t{\widetilde}

We define the pairing to be
\begin{equation} \label{pairing}
<\varphi,h>  := \int_{C \times C} \varphi (h^1 \boxtimes h^2)  - \pi i \int_{C} \Res (\varphi) h^{12} =: I_1 - I_2.
\end{equation}
The definition of $I_1$  requires some care, since the form $\varphi$ has a (second order) pole along the diagonal in $C \times C$. 

Let $u$ be a local parameter that vanishes along the diagonal $C\subset C\times C$.  
We are trying to integrate over $C\times C$ a form $\Theta$ that has a singularity near $u=0$ that looks
like $\d u\d\bar u/u^2$.  An integral with such a singularity is not absolutely convergent, so {\it a priori} it is not well-defined.
 However, because
the average of the function $1/u^2$ over a circle $|u|=\varepsilon$ (for small positive $\varepsilon$) vanishes,
the integral can be defined by restricting to $|u|\geq \varepsilon$ and then taking the limit $\varepsilon\to 0$.  Accordingly, the integral $I_1$
involves what is known in distribution theory as a principal value distribution.

A coordinate free way to define the integral $I_1$ is as follows.  Let $\Theta=\varphi(h^1\boxtimes h^2)$ be the form that we are trying
to integrate; it has the singular behavior of $\d u\d\bar u/u^2$ near the diagonal.  We make the decomposition
\begin{equation}\label{tehb}\Theta=\d\Lambda+\Xi,\end{equation}
where the form $\Xi$ is everywhere smooth and $\Lambda$ behaves near the diagonal as $\d\bar u/u$.  There is no obstruction in making such
a decomposition.  It is unique up to
\begin{equation}\label{mexo}\Lambda\to \Lambda+\lambda, ~~\Xi\to \Xi-\d\lambda,\end{equation}
where the form $\lambda$ is smooth.  Given such a decomposition, we define $I_1$ as
\begin{equation}\label{exo}I_1=\int_{C\times C}\Xi. \end{equation}
In other words, formally we take $\int_{C\times C}\d\Lambda=0$.  This definition of $I_1$ is well-defined, since $\int_{C\times C}\Xi$
is certainly invariant under $\Xi\to \Xi-\d\lambda$, where $\lambda$ is smooth.  The definition of $I_1$ given in (\ref{exo}) also agrees with the
principal value procedure mentioned in the last paragraph.  The reason for this is that if, using some arbitrary Riemannian metric on $C$, we define 
 $(C\times C)_\varepsilon$ ito be the submanifold of $C\times C$
consisting of points that are a distance $\geq \varepsilon$ from the diagonal, and if $\Lambda$ behaves near the diagonal as  $\d\bar u/u$, then
$\lim_{\varepsilon\to 0}\int_{(C\times C)_\varepsilon}\d\Lambda=0$.  So the principal value procedure would indeed instruct us, after making the decomposition
$\Theta=\d\Lambda+\Xi$, to drop the $\d\Lambda$ term, as we have done in the definition (\ref{exo}).

For future reference, we will describe this situation a little more thoroughly.  The essence of the matter is that on the complex $u$-plane $\C_u$,
the $(1,0)$-form $\d u/u^2$ makes sense as a distribution on $(0,1)$ forms with compact support.  (Equivalently, $1/u^2$ makes sense as a distribution
on smooth two-forms of compact support, or  $\d u\,\d\bar u/u^2$
makes sense as a distribution on smooth functions with compact support.)  For any $(0,1)$-form
$\chi$ with compact support, we define  the pairing of the distribution $\d u/u^2$ with $\chi$ by using the fact that $\d u/u^2=-\partial(1/u)$ where as usual
$\partial=\d u \,\partial_u$.  So formally
\begin{equation}\label{lonely} \int_{\C_u}\frac{\d u}{u^2} \chi=\int_{\C_u}\frac{1}{u}\partial\chi. \end{equation}
The integral on the right hand side is absolutely convergent, and we take this formula as the definition of the left hand side, or in other words
the definition of the pairing of the distribution $\d u/u^2$ with the $(0,1)$-form $\chi$.  
In general, the derivative of a distribution is defined by integration by parts in this fashion.  For example, if $u=u_1+iu_2$, the delta function distribution
$\delta(u)$ is defined by saying that for any smooth function $f$,
\begin{equation}\label{mexxo}\int \d u_1\d u_2 \delta(u)f(u)=f(0). \end{equation}
We note for later reference that this formula has the inconvenient-looking consequence
\begin{equation}\label{exxo}\int \d u \d \bar u \,\delta(u) f(u)=-2if(0). \end{equation}
Defining the derivative $\partial_u\delta(u)$ by integration by parts, we have for any smooth $f$
\begin{equation}\label{lexxo}\int \d u_1\d u_2 \partial_u\delta(u)f(u)=-\int\d u_1\d u_2 \delta(u)\partial_uf(u)=-\partial_uf(0). \end{equation}
Now consider the classical formula (for example, see p. 62 of \cite{FJ})
\begin{equation}\label{exxop}\frac{\partial}{\partial \bar u}\frac{1}{u}=\pi \delta(u). \end{equation}
Defining the derivative of $1/u$ (understood as a distribution on smooth two-forms) by integration by parts,
the precise meaning of this formula is that for any smooth $f$ of compact support,
\begin{equation}\label{pexxo}\int \d u_1\d u_2 \frac{1}{u}\frac{\partial f}{\partial \bar u} =-\pi f(0). \end{equation}
Differentiating the formula (\ref{exxop}) with respect to $u$, we find that
\begin{equation}\label{wexo}\frac{\partial}{\partial\bar u}\frac{1}{u^2}=-\pi\partial_u\delta(u). \end{equation}

  In the last paragraph, we have considered an isolated $\d u/u^2$ singularity at $u=0$,  but the same applies in general to a $\d u/u^2$ singularity
  along a divisor in a complex manifold.  In our application, the complex manifold is $C\times C$ and the divisor is the diagonal.  
  
\subsubsection{Gauge invariance} \label{tsdrfv3}

To show gauge invariance of our pairing, we must show that when the $h$'s vary according to \eqref{infinitesimaltransfo}, although the integrals $I_1$ and $I_2$ change,  their difference remains invariant. Because $I_1$ is bilinear in $h_1$ and $h_2$, it suffices to consider the case that they are each supported in small open sets in $C$.
Because (as the following derivation will show) gauge-dependence of $I_1$ comes only from the behavior along the diagonal, we can assume
that $h_1$ and $h_2$ are supported in the same coordinate chart $O\subset C$.  We write $z$ for a local parameter in this chart and $x,y$ for
its pullbacks to the two factors of $C\times C$.

By linearity, it suffices to check gauge invariance when
 only one of the $\alpha$'s is non-zero. We take $\alpha^1 \neq 0, \quad \alpha^2 = \alpha^{12}=0$. (The case of non-zero $\alpha^2$ is identical, and the case of non-zero $\alpha^{12}$ is trivial.) So we have:
\begin{equation}\label{ppol}
\delta h^1 = \bar{\partial}\alpha^1, \quad \delta h^2 =0, \quad  \delta h^{12} = [h^2,\alpha^1]
,\end{equation}

We write
\begin{equation}\label{bogo} 
\alpha^1=a(x,\bar x)\partial_x ,  ~~ h^2=\d\bar y \,h(y,\bar y)\partial_y,
\end{equation}
where the notation $a(x,\bar x)$ and $h(y,\bar y)$ is just meant to remind us that these functions are $C^\infty$  but neither
holomorphic nor antiholomorphic.  We also recall from \eqref{sde} that in the same local coordinates,
\begin{equation}
\varphi = f(x,y)\frac{dx^{\otimes 2} dy ^{\otimes 2} }{(x-y)^{2} },
\end{equation}
with holomorphic $f$.

We want to analyze
$$ \delta I_1=\int_{C\times C} \varphi(\bar\partial\alpha^1\boxtimes h^2) $$
by integration by  parts.   
Here,  $\bar\partial\alpha^1=\d\bar x( \partial a(x, \bar x)/\partial \bar x)\partial_x$,
and in detail our integral is in the chosen coordinates
\begin{equation} \label{zext} 
\delta I_1=\int_{C\times C}\d x\d\bar x \d y\d \bar y \, \frac{f(x,y)}{(x-y)^2} \frac{\partial}{\partial \bar x} a(x,\bar x)
\cdot h(y,\bar y).
\end{equation}
  After a few steps, we will arrive at a simple answer that
does not depend on the chosen system of coordinates.

Upon integrating by parts, we will have to calculate 
$$\frac{\partial}{\partial\bar x} \left(\frac{f(x,y)}{(x-y)^2}\right). $$
Because of the pole at $x=y$, this derivative must be understood in a distributional sense.
For a single pole, equation \eqref{exxop} gives:
\begin{equation}\label{distrib}\frac{\partial}{\partial\bar x}\frac{1}{x-y}=\pi \delta^2(x-y). \end{equation}
Differentiating this formula with respect to $y$ (or using equation \eqref{wexo}), we have:
$$\frac{\partial}{\partial\bar x}\frac{1}{(x-y)^2} 
= \pi \frac{\partial}{\partial y} \delta^2(x-y)
=-\frac{\pi}{2} (\frac{\partial}{\partial x} -\frac{\partial}{\partial y}) \delta^2(x-y). $$
Since $f$ is holomorphic, we find:
\begin{equation} \label{asd}
\frac{\partial}{\partial\bar x} \frac{f(x,y)}{(x-y)^2}=-\frac{\pi}{2}  f(x,y)  (\frac{\partial}{\partial x} -\frac{\partial}{\partial y}) \delta^2(x-y). 
\end{equation} 
A fuller explanation of the meaning of these manipulations was given
 in section \ref{tsdrfv2}.

Now we are ready to integrate by parts in the formula  (\ref{zext}):
\begin{align}
\notag \delta I_1 &= 
\frac{\pi}{2} \int_{C\times C}\d x\d \bar x\d y\d\bar y 
\left(\left(\frac{\partial}{\partial x} -\frac{\partial}{\partial y}\right) \delta^2(x-y)\right)
f(x,y)a(x,\bar x) h(y,\bar y) \\
\notag &=  \frac{\pi}{2} \int_{C\times C}\d x\d \bar x\d y\d\bar y 
\left(\left
(\frac{\partial}{\partial x} -\frac{\partial}{\partial y}\right)
\left( \delta^2(x-y)
f(x,y) \right) \right)
a(x,\bar x) h(y,\bar y) \\
\notag &=  -\frac{\pi}{2} \int_{C\times C}\d x\d \bar x\d y\d\bar y 
 \ \delta^2(x-y) f(x,y)
\left(\frac{\partial a(x,\bar x)} {\partial x}   h(y,\bar y)
-a(x,\bar x) \frac{\partial h(y,\bar y) }{\partial y}\right) \\
  \label{belmo} &=  \pi i \int_{C}\d y\d\bar y \  f(y,y)
\left(\frac{\partial a(y,\bar y)} {\partial y}   h(y,\bar y)
-a(y,\bar y) \frac{\partial h(y,\bar y) }{\partial y}\right),
\end{align}
 where:
\begin{itemize}
\item the first formula comes by  integration by parts in (\ref{zext})  with respect to $\bar x$, using \eqref{asd};
\item the second step uses the fact that   
$ \delta^2(x-y) (\frac{\partial}{\partial x} -\frac{\partial}{\partial y}) f(x,y) =0$, since $f$ is symmetric;
\item the third step is integration by parts with respect to $x$ and $y$, and
\item in the fourth step the delta function integration amounts to substituting $y$ for $x$ 
and adjusting the coefficient by a factor of $-2i$ as in \eqref{exxo}.
\end{itemize}

At this point, we can restate our result in an invariant language.  We have $\Res(\varphi)=(\d y)^2f(y,y)$, 
and the two terms in eqn. (\ref{belmo}) 
add up to a Lie bracket $[h^2,\alpha_1]$.  So
finally
$$
\begin{aligned}
\delta I_1= \pi i & \int_C \Res(\varphi)[h^2,\alpha_1]\\
= \pi i &\int_C \Res(\varphi) \delta \h \\
=\delta & I_2.
\end{aligned}
$$

\subsection{Proof for the obstruction  class} \label{sfg}

In this section, we will perform an analysis for the obstruction class $\omega_{\SM_g}$ to splitting of supermoduli space $\SM_g$ that is analogous to the analysis
of the Atiyah class  of  $\M_g$ that we have done in section \ref{oyf}.    The required steps are quite similar.  In section 
\ref{tsdrfw1}, we describe deformations
of a split super Riemann surface $S$ over the $\Z_2$-graded Artin ring $\C[\eta_1,\eta_2]$, with odd parameters $\eta_1,\eta_2$.
In section \ref{tsdrfw2}, we describe the pairing of such a deformation with an element of the extension that we claim is associated to $\omega_{\SM_g}$ and show
its gauge-invariance.  Though the logic in defining the pairing and proving its gauge invariance are the same as in the bosonic case, technically the details are simpler, since
the form we have to integrate has only a simple pole on the diagonal and the integral is absolutely convergent.
On the other hand, the description of deformations of a super Riemann surface involves a few details that are probably less familiar than their
bosonic analogs.

\subsubsection{Deformations} \label{tsdrfw1}

\def\D{{\mathcal D}}
\def\S{{\mathcal S}}
\def\h{\widehat} 
\def\O{{\mathcal O}}
We will describe deformation theory of  a super Riemann surface from a smooth point of view.
To do this, we will first  embed the sheaf of  (holomorphic) functions on a super Riemann surface $S$ as a subsheaf
of the sheaf of functions on a smooth supermanifold $\h S$. In the last sentence, the word ``holomorphic'' is in parentheses because
it is actually redundant: a super Riemann surface $S$ is best defined as a purely holomorphic object, in terms of a certain sheaf of
$\Z_2$-graded holomorphic functions (see for example \cite{Not projected} for this definition) over its reduced space $C=S_\red$, which is an ordinary
Riemann surface.  So there is no notion of functions
on $S$ that are not holomorphic.  We will introduce $\h S$ precisely in order to have such a notion.

There are several approaches to embedding the sheaf of holomorphic functions on $S$ in the sheaf of functions on a smooth supermanifold $\h S$, 
but for us the most economical approach will be most convenient.\footnote{\label{someauthors} Some authors instead take $\h S$ to be a real supermanifold of dimension
$2|2$, with local anti holomorphic coordinates $\bar z|\bar \theta$ that are complex conjugates of  local holomorphic coordinates $ z|\theta$.
 Then there is a canonical choice of $\h S$, but  on the other hand introducing $\bar\theta$ is extraneous for our purposes
 and the $\h S$ constructed this way
 is not useful in most applications in string theory.  The object $\h S$ that we describe in the text can be understood as the worldsheet of a heterotic
 string.}  We 
work in the context of what have been called cs supermanifolds in \cite{DEF}, p. 95 (see also section 5 of \cite{WittenNotes} for a more leisurely account). 
A brief definition is as follows.  A supermanifold is a locally ringed space $S=(M,\OO_S)$ consisting of a manifold $M$ and a sheaf of $\ZZ/2$-graded algebras $\OO_S$ on it which is locally isomorphic to the sheaf of sections of the exterior algebra of a vector bundle $V\to M$. 
If $M$ is a real manifold, $V$ is a real vector bundle, and by ``sections'' we mean smooth sections, then $S$ is a real supermanifold.  If $M$ is a complex manifold,
$V$ is a holomorphic vector bundle, and ``sections'' are holomorphic sections, then $S$ is a complex supermanifold.  If $M$ is real, but $V$ is a complex vector bundle,
and ``sections'' are complex-valued smooth sections, then $S$ is a cs supermanifold.

  Like a real supermanifold, a cs supermanifold $S$ is determined
up to isomorphism by its reduced space $M$ and the normal bundle to $M$ in $S$, which is also called the odd tangent bundle $T_-M$.   This can be proved by the same arguments used for real supermanifolds. In the case of a real
supermanifold, $T_-M=\Pi V$ where $V\to M$ is a real vector bundle  (and $\Pi V$ is $V$ with parity reversed, in other words, with the fibers  taken to be purely odd); 
in the case of a cs supermanifold, $T_-M=\Pi V$ where now $V\to M$ is a complex vector bundle.  From the fact that $S$ is determined up to isomorphism
by the pair $M, V$, it follows that any family of cs super manifolds is locally constant, up to isomorphism.  
There is no notion of a function on a cs supermanifold being real (except after reducing modulo the odd variables),
just as there is no notion of a real section of a complex vector bundle.
Functions on a cs manifold are the analog of complex-valued smooth functions on an ordinary manifold.   
A real supermanifold is the same thing as a cs supermanifold
in which one is given a notion of what functions are real. 
  
A split complex supermanifold $S=(M,\bigwedge^{\bullet}V)$ determines the cs supermanifold
\begin{equation}\label{orelf}
\h S := (M_{\RR},{\mathcal C}^{\infty}(M_{\RR}) \otimes_{\OO_{M}} \OO_S) =  
(M_{\RR},{\mathcal C}^{\infty}(M_{\RR}) \otimes_{\OO_{M}}(\bigwedge^{\bullet}V))  =
(M_{\RR},{\mathcal C}^{\infty}(\bigwedge^{\bullet}V)), 
\end{equation}
where $M_{\RR}$ is the real smooth manifold underlying the complex $M$. This still works when $S$ is projected, but does not extend naturally to the general nonsplit case
(in which $S$ depends on some odd parameters), since in general $\OO_S$ need not be a sheaf of $\OO_M$ modules.

For our application,
we take $\h S$ to be a smooth cs supermanifold of dimension $2|1$, whose reduced space is an ordinary Riemann surface $C$.
   We describe $S$ locally by even coordinates
$\t z$ and $z$ and an odd coordinate $\theta$.  We assume\footnote{This assumption is stronger than necessary: in what follows, it would make sense to   assume
merely that $\t z$ is sufficiently close to the complex conjugate of $z$ in the sense that 
 $\t z$ and $z$ reduce to  local complex coordinates in oppositely oriented complex structures on $\hat S_\red$.
 However, since we will work over an Artin ring, and make only nilpotent deformations away from the familiar situation in which $\t z$ is
the complex conjugate of $z$, the condition stated in the text is natural.}  that modulo the odd variables, $\t z$ is the complex conjugate of $z$
and that $z$ and $\t z$ are local holomorphic and antiholomorphic coordinates on the reduced space $S_\red$.  We assume that $\h S$
is endowed with a sheaf of holomorphic functions, which are the functions annihilated by an operator that we can write
\begin{equation}\label{melob}\t\partial=\d\t z \frac{\partial}{\partial\t z}. \end{equation}
Thus locally holomorphic functions are functions $f(z|\theta)$. 
The reduced space of $\h S$ is assumed to be that of the purely holomorphic super Riemann surface $S$ that we started with (namely
the Riemann surface $C$), and the sheaf of
holomorphic functions on $\h S$, understood as a sheaf on $C$, is assumed to coincide with the sheaf of holomorphic functions on $S$.
Thus one can think of $\h S$ as a smooth supermanifold that maps to the complex supermanifold $S$.
 Similarly, $\h S$ is endowed with a sheaf of antiholomorphic
functions which locally are  functions of $\t z$, in other words the functions annihilated by $\partial_z$ and $\partial_\theta$.
(These vector fields generate the holomorphic tangent bundle of $\h S$, which we will denote as  $TS$, since it corresponds in an obvious sense to
the tangent bundle of the super Riemann surface $S$.  A section of $TS\to \h S$ is holomorphic if  it can be written locally as $\nu \partial_\theta +w\partial_z$
where $\nu$ and $w$ are holomorphic functions on $\h S$. )
Thus one can also think of $\h S$ as a smooth supermanifold that maps to a complex manifold $\t{C}$ (which is the complex conjugate of  $S_\red=C$ if $z$ reduces
modulo nilpotents to the complex conjugate of $\t z$). Putting these together, we get an embedding of  $\h S$  into 
$S \times \t{C}$.
The choice of the embedding is not canonical (unless $S$ is split),
but varying the embedding (with small parameters) affects neither the holomorphic
nor the antiholomorphic structure of $\h S$, which are pulled back from $S$ and $\t{C}$, respectively.
Deformations of the holomorphic and antiholomorphic structures of $\h S$ are therefore
simply deformations of the holomorphic structures of  $S$ and $\t{C}$.

\def\K{{\mathcal K}}
   The subtlety in this construction is that unless $S$ is split, there is no canonical construction of
$\h S$.   
(Our ultimate application will involve the case that $S$ is split, in which case, as in eqn. (\ref{orelf}),  there is a canonical $\h S$.)
There is no obstruction to deforming $\h S$ to compensate for small
deformations of $S$ (or of $\t C$), so a suitable $\h S$ always exists.   Any choice of $\h S$ provides a framework for studying deformations of $S$
from a smooth point of view.
Moreover, since a family of cs supermanifolds is locally constant up to isomorphism, when we vary $\t C\times S$ over an Artin ring
(or even when we make small deformations in the complex topology) we can consider $\h S$ to be constant while only its  holomorphic and
antiholomorphic structures are modified.   This is the analog of the statement that in deformation theory of an ordinary Riemann surface $C$, we can consider
$C$ to be fixed as a smooth two-manifold, while only its holomorphic and antiholomorphic structures vary.  So it is the reason that the consideration of $\h S$
gives a framework for studying the deformations of $S$ from a smooth point of view.

What we have described so far is a cs supermanifold $\h S$ 
 that has a holomorphic structure in which it is a  complex supermanifold of dimension $1|1$, and an antiholomorphic structure in which it is
 a complex supermanifold of dimension $1|0$.  (Moreover, these structures are complex conjugate if one reduces modulo the odd variables.)  $\h S$ is a smooth model of a super 
Riemann surface if in addition $TS$ is endowed with a rank $0|1$ holomorphic  subbundle  $\D$ that is everywhere unintegrable.  An easy lemma shows that if so, then locally
one can choose local holomorphic coordinates $z|\theta $ on $\h S$ -- called local superconformal coordinates -- such that  $\D$ is generated by
\begin{equation}\label{welob} D_\theta=\partial_\theta+\theta\partial_z. \end{equation}
The sheaf of  holomorphic sections of $TS$ has a subsheaf $\S$ consisting of ``superconformal vector fields.'' These are holomorphic vector fields
that preserve the subbundle $\D\subset TS$.  Concretely, in local superconformal coordinates $z|\theta$, the general form of 
a superconformal vector field $u$ is
\begin{equation}\label{below}u= \nu(z)\left(\partial_\theta-\theta\partial_z\right)+ w(z)\partial_z+\frac{1}{2}\partial_z w(z)\theta\partial_\theta,\end{equation}
where $\nu(z)$ and $w(z)$ are functions only of $z$ and not $\theta$.  If $S$ is split (but not otherwise), the decomposition of $u$ as the sum of
an ``even'' part involving $w$ and an ``odd'' part proportional to $\nu$ is valid globally. Locally, a  smooth section of $\S\to \h S$ is given again by the formula (\ref{below}),
except that now $\nu$ and $w$ depend on both $z $ and $\t z$ (but not on $\theta$), as in eqn. (\ref{melf}) below.

Now let us recall some facts about the purely holomorphic theory of a super Riemann surface $S$.
The first order deformations of $S$ as a complex supermanifold are parametrized by $H^1(S,TS)$, but its first order deformations
as a super Riemann surface are parametrized by $H^1(S,\S)$.  The basis for these statements is standard:  $TS$ (or $\S$) is the sheaf of
infinitesimal automorphisms of $S$ as a complex supermanifold (or as a super Riemann surface), so $H^1(S,TS)$ (or $H^1(S,\S)$) parametrizes
its first order deformations as a complex supermanifold (or as a super Riemann surface).    In all these statements, $S$ is understood as a purely
holomorphic object, and $H^1(S,TS)$ and $H^1(S,\S)$ can be defined, for example, as Cech cohomology groups. 

An immediate consequence of introducing $\h S$ is that there is a Dolbeault model for $H^1(S,TS)$ and $H^1(S,\S)$:  they can be computed
as the Dolbeault cohomology groups $H^1(\h S,TS)$ and $H^1(\h S,\S)$.   (This can be proved by slightly adapting any standard proof of the relation of Cech and Dolbeault cohomology.)
Thus, an element of $H^1(\h S,TS)$ is represented by a $(0,1)$-form $\lambda$ on $\h S$ valued in $TS$.  Locally, after trivializing $TS$ with the basis $\partial_\theta,\partial_z$,
we can write such a form as
\begin{equation}\label{mendo}\lambda=\chi\partial_\theta +h\partial_z,\end{equation}
where $\chi$ and $h$ are smooth 
(0,1)-forms on $\h S$.
(One cannot make such a decomposition of $\lambda$ globally unless $S$ is split.)
The deformation by $\lambda$ is subject to the usual equivalence relation of Dolbeault cohomology
\begin{equation}\label{endo}\lambda\to \lambda+\t\partial u, \end{equation}
where $u$ is a smooth section of $TS\to \h S$.  The advantage of the Dolbeault description is that it is straighforward to describe deformations of higher order.    
To go beyond a first order deformation, we perturb the operator $\t\partial$ to 
\begin{equation}\label{mody} \t\partial'=\t\partial+\chi \frac{\partial}{\partial\theta}+h\frac{\partial}{\partial z},
\end{equation} 
and we do not work just to first order in $\chi$ and $h$.  After the perturbation, a holomorphic function is a smooth function on $\h S$ that is annihilated by
$\t\partial'$.
For example, to construct deformations over a $\Z_2$-graded Artin ring $B$, we simply take $\chi$ and $h$ to have coefficients in $B$.
The gauge equivalence relation for such deformations is simply
\begin{equation}\label{bundle}\t\partial' \to e^{-u}\t\partial e^u, \end{equation}
where $u$ is a section of $TS$ with coefficients in $B$.  Here we take $\chi, h$, and $u$ to all vanish if reduced modulo the ideal of nilpotent elements of $B$.

Similarly, $H^1(S,\S)$, which parametrizes the first order deformations of $S$ as a super Riemann surface, can be computed as the Dolbeault cohomology group
$H^1(\h S,\S)$.   The procedure is the same as before, except that  $\lambda$ must now be a $(0,1)$-form valued in $\S$:
\begin{equation}\label{pendo}\lambda=\chi(z,\t z)\left(\partial_\theta-\theta\partial_z\right)+\left(h(z,\t z)\partial_z 
+\frac{1}{2}\partial_z h(z,\t z)\theta\partial_\theta\right) .\end{equation}
The equivalence relation can still be written as in (\ref{endo}), but now $u$ is a smooth section of $\S$:
\begin{equation}\label{melf}u=\nu(z,\t z)\left(\partial_\theta-\theta\partial_z\right)+w(z,\t z)\partial_z+\frac{1}{2}\partial_z w(z,\t z)\theta\partial_\theta. \end{equation}
To go beyond a first order deformation, we perturb the operator $\t\partial$ to
\begin{equation}\label{ody} \t\partial'=\t\partial+\chi(z,\t z)\left(\partial_\theta-\theta\partial_z\right)
+\left(h(z,\t z)\partial_z+\frac{1}{2}\partial_z h\theta\partial_\theta\right).\end{equation}
The equivalence relation takes the same form as in (\ref{bundle}) except that now, of course, $u$ is a smooth section of $\S$ and thus has the local
form (\ref{melf}).

The procedure of the last paragraph can be used, in particular,  to describe deformations of the super Riemann surface $S$ parametrized by
any $\Z_2$-graded Artin ring $B$.  For our purposes, we want to take $B$ to be $\C[\eta_1,\eta_2]$, where $\eta_1$ and $\eta_2$ are odd parameters.  
We also assume that prior to the deformation, $S=\Pi T^{1/2}C$ is split.  (This is the case we need to investigate the obstruction class $\omega_{\SM_g}$.)
Then the most general 
nilpotent
deformation (\ref{ody}) takes the form
\begin{equation}\label{blocy}\t\partial'=\t\partial+\sum_{i=1,2}\eta_i\chi^i(z,\t z)\left(\partial_\theta-\theta\partial_z\right)
+\eta_1\eta_2\left(h^{12}(z,\t z)\partial_z+\frac{1}{2}\partial_z h^{12}(z,\t z)\theta\partial_\theta\right).\end{equation}
Thus, the deformation involves the two $(0,1)$-forms $\chi^i$, $i=1,2$ on $C$ valued in $T^{1/2}C$, and a $(0,1)$-form $h^{12}$ on $C$ valued in $TC$.
The gauge parameter $u$ has a similar expansion:
\begin{equation}\label{locy} u =\sum_{i=1,2} \eta_i \nu^i(z,\t z)\left(\partial_\theta-\theta\partial_z\right) +\eta_1\eta_2 \left(w(z,\t z)\partial_z +\frac{1}{2}
\partial_z w(z,\t z)\theta\partial_\theta\right). \end{equation}
The equivalence relation (\ref{bundle}) can easily be worked out explicitly.  As in the bosonic case, the most useful version of the
equivalence relation is the linearized version, in which we work to first order in the vector field $u$.
This comes out to be \begin{align}\label{delbox} \chi^i& \to \chi^i+\t\partial \nu^i~~i=1,2\cr
                                              h^{12} & \to h^{12}+\t\partial w +\nu^1\chi^2-\nu^2\chi^1. \end{align}
The analogy with the corresponding bosonic formula (\ref{infinitesimaltransfo}) is hopefully clear.

\subsubsection{The pairing and its gauge invariance} \label{tsdrfw2}

We are in a situation very similar to what we encountered in section \ref{tsdrfv2}.  
The claim of  Proposition \ref{CxC} is that the extension expressing the first
obstruction $\omega$ to splitting of $\SM_g$ is associated to the global sections on $C \times C$ of 
$\OO(3,3,1)= K^{3/2}\boxtimes K^{3/2}\otimes\O(\Delta)$.
We view a section $\varphi$ of $\OO(3,3,1)=K^{3/2}\boxtimes K^{3/2}\otimes \O(\Delta)$ as a section of $K^{3/2}\boxtimes K^{3/2}$ that has
a pole along the diagonal $C\subset C\times C$.  The residue of this pole is a section $\Res(\varphi)$ of $K^2\to C$.

To establish the proposition, we need to show that there is a natural pairing between 
such a $\varphi$  and the triple
$h=(\chi^1,\chi^2, h^{12})$ that describe a deformation of $S$ over the ring $\C[\eta_1,\eta_2]$.
Clearly, if $\Res(\varphi)=0$,  so that $\varphi$ is a holomorphic section of $K^{3/2}\boxtimes K^{3/2}\to C\times C$, then $\varphi$ can be paired with
$\chi^1\boxtimes\chi^2\in H^2(C\times C,T^{1/2}\boxtimes T^{1/2})$.  Equally clearly, $\Res(\varphi)\in H^0(C, K^2)$ can be paired with $h^{12}\in H^1(C, T)$.  

The more subtle fact is that these pairings can be combined to give a pairing of $\varphi$ with the triple $h=(\chi^1,\chi^2, h^{12})$ that is invariant under (\ref{delbox}).  
The definition is precisely analogous to eqn. (\ref{pairing}):
\begin{equation}\label{fermipairing}\langle \varphi,h\rangle:= \int_{C\times C}\varphi(\chi^1\boxtimes\chi^2) - 2\pi i \int_C\Res(\varphi)h^{12}:=I_1- I_2. \end{equation}

Because $\varphi$ now has only a simple pole along the diagonal, there is not much to say about the well-definedness of this formula: the integral $I_1$
is absolutely convergent, because the form $ \d u\d\bar u/u$ is integrable near $u=0$.  The proof of gauge-invariance is similar to what
we explained in section \ref{tsdrfv3}, but much simpler, again because $\varphi$ has only a simple pole.  By linearity, it suffices to verify separately the 
invariance of the pairing under the transformations generated by the parameters $\nu^1$, $\nu^2$, and $w$ in (\ref{delbox}).  Invariance under the even
parameter $w$ is a
triviality, and the odd parameters $\nu^1$ and $\nu^2$ enter symmetrically (they are exchanged by the exchange of the two factors of $C\times C$),
so it suffices to verify that the pairing is invariant under the transformation generated by $\nu^1$.

We have
\begin{equation}\label{dolf}\delta I_1=\int_{C\times C}\varphi(\t\partial \nu^1,\chi^2). \end{equation}
Integrating by parts as we did in \eqref{belmo} but using  (\ref{distrib}) instead of (\ref{asd}), we get
\begin{equation}\label{olf}\delta I_1=2\pi i\int_C\Res(\varphi) \nu^1\chi^2. \end{equation}
In this formula, the product $\nu^1\chi^2$ is a $(0,1)$-form valued in $T$, which can be naturally paired with $\Res(\varphi)\in H^0(C,K^2)$.
In (\ref{fermipairing}), this variation of $I_1$ cancels the variation of $I_2$ that comes from the $\nu^1\chi^2$ term in the transformation of
$h^{12}$.

\subsection{Variant for punctured SRSs} \label{puncture}

\mbox {   }

Proposition \ref{CxC} has an analogue giving an extension class interpretation 
of the first obstruction to the splitting of the moduli space $\SM_{g,1}$ of marked super Riemann surfaces. 
We fix the point   $(S,p)$ of $\SM_{g,1}$, where $S$ is a split super Riemann surface:
$S=S(C,T_C^{1/2})$, with puncture $p \in C$.
Consider the sheaves (line bundles, actually):
\begin{equation}\label{abcp} 
\OO(a(p),b(p),c) := {{p_1}^* {K_C}^{\otimes a/2} (p)} \otimes  { {p_2}^* {K_C}^{\otimes b/2}(p) }    
\otimes   {\OO_{C \times C} (c \Delta)},
\end{equation}
and especially $\OO(3(p),3(p),1)$. As before, the restriction to the diagonal gives a short exact sequence:
\begin{equation}\label{331p}
0 \to \OO(3(p),3(p),0) \to \OO(3(p),3(p),1) \to  K_C{^2}(2p) \to 0.
\end{equation}
This may be the most obvious guess for the punctured analogue of Proposition \ref{CxC},
but it turns out to give the wrong answer. Instead, we have to introduce the subsheaves
\begin{equation}\label{abcIp} 
I_p(a(p),b(p),c):=I_p \otimes  \OO(a(p),b(p),c),
\end{equation}
where $I_p$ is the ideal sheaf of the point $(p,p)$ in $C \times C$.  Now the 
restriction to the diagonal gives a short exact sequence:
\begin{equation}\label{331Ip}
0 \to I_p(3(p),3(p),0) \to I_p(3(p),3(p),1) \to  K_C{^2}(p) \to 0
\end{equation}
where each term is a subsheaf of the corresponding term in \eqref{331p}. The long exact sequence of \eqref{331Ip} gives:
\[
0 \to H^0(  C{\times}C,  I_p(3(p),3(p),0) )\to H^0(  C{\times}C,  I_p(3(p),3(p),1) )  \to H^0(C, {K_C{^2}(p))}  \to 0,
\]
and the even part under the involution  $(x,y) \to (y,x)$ gives:
\begin{equation}\label{CxC with puncture} 
0 \to \bigwedge^2 H^0(K_{C}^{3/2}(p))
\to H^0(C{\times}C, \mbox{  }  I_p(3(p),3(p),1) \mbox{ } )^{+}
\to H^0(K_{C}^{2}(p))
\to 0.
\end {equation}
As previously, we are giving only the fiber  of this SES at the point   $(C,T_C^{1/2},p)$ of $\SM_{g,1}$. 
Everything globalizes naturally to a SES of vector bundles on $\SM_{g,1}$, 
obtained by replacing each occurrance of $H^0$ by $\pi_*$, where
$\pi: \SM_{g,1} \to \SM_{g}$ is the natural projection (= the universal curve).
\begin{proposition}\label{poiuy}
The extension class of SES \eqref{CxC with puncture} is the first obstruction to the splitting of  the moduli space  $\SM_{g,1}$ of punctured SRS.
\end{proposition}

The proof is of course parallel to the proof of Proposition \ref{CxC} given in the previous section.
Instead of a split super Riemann surface $S$ with underlying Riemann surface $C$, 
our basic object is now the punctured super Riemann surface: $(S,p)$.
Its sheaf of infinitesimal automorphisms as a punctured supermanifold is\footnote{$p$ is not a divisor in $S$ but a subvariety of codimension $1|1$,
so the sheaf of functions on $S$ that vanish at $p$ is not locally free, and either is $T_S(-p)$.  In the super Riemann surface case that we come to momentarily,
both $\S$ and $\S(-p)$ can be given the structure of a locally free sheaf, but we will not make use of this.}  $T_S(-p)$, 
the subsheaf of $T_S$ consisting of vector fields that vanish at $p$.
Likewise, its sheaf of infinitesimal automorphisms as a punctured super Riemann surface is  $\S(-p)$, 
the subsheaf of $\S$ consisting of superconformal vector fields that vanish at $p$.
It follows that the first order deformations of $(S,p)$ as a punctured complex supermanifold 
are parametrized by $H^1(S,TS(-p))$, 
while its first order deformations as a punctured super Riemann surface 
are parametrized by $H^1(S,\S(-p))$. 
As previously, this can be computed as the Dolbeault cohomology group
$H^1(\h S,\S(-p))$.  
The procedure is the same as in \eqref{pendo}, except that 
the  $(0,1)$-form $\lambda$ valued in $\S$: 
\begin{equation}\label{pendop}\lambda=\chi(z,\t z)\left(\partial_\theta-\theta\partial_z\right)+\left(h(z,\t z)\partial_z 
+\frac{1}{2}\partial_z h(z,\t z)\theta\partial_\theta\right) .\end{equation}
must now vanish at the puncture $p$.  Note tht $p$ is given by a condition such as $z|\theta=z_0|\theta_0$, and the condition
for $\lambda$ to vanish at $p$ is concretely 
\begin{equation}\label{holmo}-\chi(z_0,\t z)\theta_0+h(z_0,\t z)=0,~~\chi(z_0,\t z)+\frac{1}{2}\partial_z h(z_0,\t z)\theta_0=0.\end{equation}
The equivalence relation can also still be written as in (\ref{melf}), except that the gauge parameter $u$: 
\begin{equation}\label{melfp}u=\nu(z,\t z)\left(\partial_\theta-\theta\partial_z\right)+w(z,\t z)\partial_z+\frac{1}{2}\partial_z w(z,\t z)\theta\partial_\theta \end{equation}
must now vanish at $p$ as well.
In particular, the deformations parametrized by the $\Z_2$-graded Artinian ring $B=\C[\eta_1,\eta_2]$
are given by the immediate analogue of \eqref{blocy}:
\begin{equation}\label{blocyp}\t\partial'=\t\partial+\sum_{i=1,2}\eta_i\chi^i(z,\t z)\left(\partial_\theta-\theta\partial_z\right)
+\eta_1\eta_2\left(h^{12}(z,\t z)\partial_z+\frac{1}{2}\partial_z h^{12}(z,\t z)\theta\partial_\theta\right),\end{equation}
with coefficients that vanish at $p$. 
Thus, the deformation involves
\begin{equation} \label{btfg}
h^{12} \in {\mathcal{A}}^{0,1}(T_C(-p)), \quad
\chi^i \in {\mathcal{A}}^{0,1}(T_C^{1/2}(-p)), \ i=1,2.
\end{equation}
The gauge parameter $u$ has a similar expansion:
\begin{equation}\label{locyp} u =\sum_{i=1,2} \eta_i \nu^i(z,\t z)\left(\partial_\theta-\theta\partial_z\right) +\eta_1\eta_2 \left(w(z,\t z)\partial_z +\frac{1}{2}
\partial_z w(z,\t z)\theta\partial_\theta\right), \end{equation}
where now $\nu^i,w$ must vanish at $p$.
The (linearized) equivalence relation remains precisely (\ref{delbox}):
\begin{align}\label{delboxp} \chi^i& \to \chi^i+\t\partial \nu^i~~i=1,2\cr
  h^{12} & \to h^{12}+\t\partial w +\nu^1\chi^2-\nu^2\chi^1. \end{align}
  In our application, we are interested in deformations of a {\it split}  punctured super Riemann surface $S$ over the Artin ring
  $B=\C[\eta_1,\eta_2]$, so in contrast to eqn. (\ref{holmo}), we set $\theta_0=0$,  and view $p$ as a point in the reduced space $C=S_\red$.  So
  the $\chi^i$ are $(0,1)$-forms on $C$ valued in $TC^{1/2}(-p)$ and $h^{12}$ is a $(0,1)$-form on $C$ valued in $TC(-p)$. Likewise $\nu^i$ and $w$
  are smooth sections of $TC^{1/2}(-p)$ and of $TC(-p)$.

As in previous sections, we will now describe a pairing between even sections 
 $\varphi$ of $I_p(3(p),3(p),1)$ and the second order deformation data.
We view a section $\varphi$ of 
$\OO(3(p),3(p),1)=K^{3/2}(p)\boxtimes K^{3/2}(p)\otimes \O(\Delta)$ 
as a section of $K^{3/2}\boxtimes K^{3/2}$ that has
a pole along the diagonal $C\subset C\times C$ 
and along the horizontal and vertical curves $C \times p, p \times C$.  
The residue along the diagonal is a section $\Res(\varphi)$ of $K^2(2p)$.

Let $z$ be a local coordinate on $C$ centered at $p$, 
and let $x:=p_1^*(z), y:=p_2^*(z)$ be the corresponding local coordinates on  $C \times C$. 
Locally, our section $\varphi$ is of the form 
\begin{equation} \label{explicitp}
\varphi=f(x,y)\frac{\d x^{\otimes \frac{3}{2}} \d y ^{\otimes \frac{3}{2}}} {xy(x-y) },
\end{equation}
and its residue is:
\[
\Res(\varphi) = \frac{f(z,z)}{z^2} \d z^2.
\]
Note that our $\varphi$ is a section of the subsheaf $I_p(3(p),3(p),1) \subset \OO(3(p),3(p),1)$ if and only if 
$f(0,0)=0$, which is exactly the condition for $\Res(\varphi)$ to land in the subsheaf  
$K^2(p) \subset  K^2(2p)$. 
Note also that the function $f(x,y)$ is even (or odd) if and only if 
$\varphi$  is even (or odd, respectively), with respect to the action of the involution
 $(x,y) \to (y,x)$ of $C \times C$.

To establish the proposition, we need to show that there is a natural pairing between 
such a $\varphi$  and the triple
$h=(\chi_1,\chi_2, h_{12})$, as in  \eqref{btfg}, which describes a deformation of $(S,p)$ over the ring $\C[\eta_1,\eta_2]$.
Clearly, if $\Res(\varphi)=0$,  
so that $\varphi$ is a holomorphic section on $C\times C$ of $K^{3/2}(p)\boxtimes K^{3/2}(p)$, 
then $\varphi$ can be paired with
$\chi^1\boxtimes\chi^2\in H^2(C\times C,T^{1/2}(-p)\boxtimes T^{1/2}(-p))$.  
Equally clearly, $\Res(\varphi)\in H^0(C, K^2(p))$ can be paired with $h^{12}\in H^1(C, T(-p))$.  

As previously, we need to check that these pairings can be combined to give a pairing of $\varphi$ with the triple $h=(\chi^1,\chi^2, h^{12})$ that is invariant under (\ref{delboxp}). 
The definition is identical to equation (\ref{fermipairing}):
\begin{equation}\label{fermipairingp}\langle \varphi,h\rangle:= \int_{C\times C}\varphi(\chi^1\boxtimes\chi^2) +2\pi \int_C\Res(\varphi)h^{12}:=I_1+I_2. \end{equation}
The extra poles of $\varphi$, along the horizontal and vertical copies of $C$, 
are canceled by the extra zeros of $\chi^1,\chi^2$. So our new integral $I_1$ has the same behavior
along the diagonal as the one in 
equation (\ref{fermipairing}), and is therefore absolutely convergent. 
Likewise, the proof of gauge-invariance is identical to the non-punctured case, since the transformation formulas \eqref{delboxp} have the same form as the transformation formulas \eqref{delbox}, 
and are applied to a subset of the parameters $w, \nu^i$ that appeared there.
This completes the proof of Proposition \ref{poiuy}.

\section{Non-projectedness of $\SM_{g,1}$} \label{punctured non-splitness} \label{4}

The main result of \cite{Not projected} was that super moduli space $\SM_g$ 
is non-projected and non-split for $g \geq 5$. 
As explained there, this follows  from an algebraic geometry construction 
(embedding  $\SM_{g,1}$ in  $\SM_{\widetilde{g}}$ for appropriate ${\widetilde{g}}$)
together with:

\begin{theorem}
The first obstruction to the splitting of $\SM_{g,1}$:
\[
\omega \in H^1(\sM_{g,1}, \Hom(\exterior{2}T_-,T_+))
\]
does not vanish for $g \geq 2$ (and even spin), so the supermanifold $\SM_{g,1}$ is non-projected. 
\end{theorem}

Here and in the rest of this section, $T_{\pm}$ refer to  $T_{\pm}\SM_{g,1}$, which  are vector bundles on the reduced space $\sM_{g,1}$. 
In \cite{Not projected} we gave a proof of this result which relied on somewhat delicate properties of supermanifolds and maps between them. We reprove it here using the deformation theory we have developed, especially the explicit cohomological interpretation of the obstruction class.

\begin{proof}[proof]

We fix  a spin curve $(C,T_C^{1/2}) \in \sM_g^+$ which we identify as  a fiber of $\pi: \SM_{g,1}  \to \SM_g$, and choose a point $p \in C \subset \sM_{g,1}^+$. 
We have interpreted 
\[
\omega \in H^1(  \sM_{g,1}^+   , \Hom(\exterior{2}T_-,T_+))
\]
as the extension class of a short exact sequence of vector bundles, and have obtained a canonical identification \eqref{CxC with puncture} of the fiber at  $(C,T_C^{1/2},p)$  of this SES with:
\begin{equation}\label{with p}
0 \to \bigwedge^2 H^0(K_{C}^{3/2}(p))
\to H^0(C{\times}C, \mbox{  }  I_p(3(p),3(p),1) \mbox{ } )^{+}
\to H^0(K_{C}^{2}(p))
\to 0.
\end {equation}
On the other hand, we also have the pullback to $\sM_{g,1}$ of the first obstruction to the splitting of $\SM_{g}$. In Proposition {\ref{CxC}} we have obtained a canonical identification of the fiber at $(C,p)$ of the corresponding SES with our SES {\eqref{LES331}}$^+$:
\begin {equation}\label{without p}
0 \to \bigwedge^2 H^0(K_{C}^{3/2})
\to H^0(C{\times}C, \mbox{  } \OO(3,3,1) \mbox{ } )^{+}
\to H^0(K_{C}^{2})
\to 0.
\end {equation}
The pullback $\pi^*$ maps sequence \eqref{without p} into sequence \eqref{with p} , so we form the term by term quotient sequence $Q$:
\[
0 \to Q_1  \to Q_2 \to Q_3 \to 0.
\]
on the other hand, consider the odd cotangent sequence of $\pi$:
\begin{equation}\label{rel odd cotan}
0 \to \pi^* T_{-}^* \SM_g \to 
T_{-}^*\SM_{g,1}       \to 
T_{-}^*\pi  \to 0,
\end{equation}
whose fiber at $(C,T_C^{1/2},p)$ is:
\begin{equation}\label{rel odd cotan explicit}
0 \to       H^0(K_{C}^{3/2})     \to 
H^0(K_{C}^{3/2}(p))      \to 
 K^{1/2}_{C,p}       \to 0.
\end {equation}

\begin{lemma}
The quotient sequence $Q$ is canonically identified with the odd cotangent sequence \eqref{rel odd cotan} of $\pi$, tensored with the odd relative cotangent line bundle $T_{\pi,-}^*={T^*_{\SM_{g,1}  / \SM_g ,-}} $, whose fiber at $(C,T_C^{1/2},p)$ is \[
 (T_{\pi,-}^*) _{  (C,T_C^{1/2},p)    }  \quad =  \quad (T_{(C,T_C^{1/2},p),-}^*\SM_{g,1})/( T_{(C,T_C^{1/2},-)}^*\SM_{g}) \quad =  \quad      K_{C,p}^{1/2}.
\]
\end{lemma}
\begin{proof}
We are trying to identify the bottom row of the array:

\begin{tikzpicture}[auto]
  \node (A) {$\bigwedge^2 H^0(K_{C}^{3/2})$};
  \node (B) [below of=A] [node distance=2cm]{$\bigwedge^2 H^0(K_{C}^{3/2}(p))$};
  \node (C) [below of=B] [node distance=2cm]{$Q_1$};
  \node (A1) [right of=A][node distance=5cm] {$H^0(C{\times}C, \mbox{  } \OO(3,3,1) \mbox{ } )^{-}$};
  \node (B1) [right of=B] [node distance=5cm]{$H^0(C{\times}C, \mbox{  }  I_p(3(p),3(p),1) \mbox{ } )^{-}$};
  \node (C1) [right of=C] [node distance=5cm] {$Q_2$};
  \node (A2) [right of=A1][node distance=5cm] {$H^0(K_{C}^{2})$};
  \node (B2) [right of=B1] [node distance=5cm]{$H^0(K_{C}^{2}(p))$};
  \node (C2) [right of=C1] [node distance=5cm] {$Q_3$};
  \draw[->] (A) to node {} (B);
  \draw[->] (B) to node {} (C);
  \draw[->] (A1) to node {} (B1);
  \draw[->] (B1) to node {} (C1);
  \draw[->] (A2) to node {} (B2);
  \draw[->] (B2) to node {} (C2);
  \draw[->] (A) to node [swap] {} (A1);
  \draw[->] (B) to node {} (B1);
  \draw[->] (C) to node {} (C1);
  \draw[->] (A1) to node [swap] {} (A2);
  \draw[->] (B1) to node {} (B2);
  \draw[->] (C1) to node {} (C2);
\end{tikzpicture}

in which the top two rows are sequences \eqref{without p} and \eqref{with p} respectively. Each of the nine entries here represents a vector bundle over $\sM_g^+$, but for simplicity we will focus on its fiber at the point $(C,T_C^{1/2},p)$, which is a vector space (depending of course on $(C,T_C^{1/2},p)$).  The point is that we have already identified each of the top six of these vector spaces, as well as the maps between them, as global sections  of an appropriate sheaf over $C \times C$ and the induced maps between them. The corresponding array of sheaves on  $C \times C$ is:

\begin{tikzpicture}[auto]
  \node (A) {$\OO(3,3,0)^{-}$};
  \node (B) [below of=A] [node distance=2cm]{$I_p(3(p),3(p),0)^{-}$};
  \node (C) [below of=B] [node distance=2cm]{$K_C^{3/2} \otimes K_{C,p}^{1/2} $};
  \node (A1) [right of=A][node distance=5cm] {$\OO(3,3,1)^{-}$};
  \node (B1) [right of=B] [node distance=5cm]{$I_p(3(p),3(p),1)^{-}$};
  \node (C1) [right of=C] [node distance=5cm] {$K_C^{3/2}(p) \otimes K_{C,p}^{1/2} $};
  \node (A2) [right of=A1][node distance=5cm] {$K_{C}^2$};
  \node (B2) [right of=B1] [node distance=5cm]{$K_{C}^2(p)$};
  \node (C2) [right of=C1] [node distance=5cm] {$K_{C,p}^{1/2}  \otimes  K_{C,p}^{1/2} $};
  \draw[->] (A) to node {} (B);
  \draw[->] (B) to node {} (C);
  \draw[->] (A1) to node {} (B1);
  \draw[->] (B1) to node {} (C1);
  \draw[->] (A2) to node {} (B2);
  \draw[->] (B2) to node {} (C2);
  \draw[->] (A) to node [swap] {} (A1);
  \draw[->] (B) to node {} (B1);
  \draw[->] (C) to node {} (C1);
  \draw[->] (A1) to node [swap] {} (A2);
  \draw[->] (B1) to node {} (B2);
  \draw[->] (C1) to node {} (C2);
\end{tikzpicture}

Here each row is a restriction of a sheaf on $C \times C$ to the diagonal $\Delta$, while each column is a restriction of a sheaf on $C \times C$ to the horizontal curve $C \times p$. The sequence in our Lemma is now obtained as global sections of the bottom row.
\end{proof}
It may be helpful to describe sections of the top left corner of the last diagram in terms of local coordinates $x,y$ on the two copies of $C$ near $p$. A local section of $\OO(3,3,0)^{}$ can be written as $a(x,y) {dx}^{3/2} \otimes  {dy}^{3/2}$ for some holomorphic function $a$. A local section of $\OO(3,3,1)^{}$ can likewise  be written as $(\frac{a(x,y)}{x-y} )     {dx}^{3/2} \otimes  {dy}^{3/2}$ for some holomorphic function $a$.  A local section of $I_p(3(p),3(p),0)^{}$ can be written as $(\frac{a(x,y)}{x}           - \frac{b(y,x)}{y})  {dx}^{3/2} \otimes  {dy}^{3/2}$ for some holomorphic functions $a,b$. And a  local section of $I_p(3(p),3(p),1)^{}$ can be written as $((\frac{a(x,y)}{x}  - \frac{b(y,x)}{y}) /(x-y))  {dx}^{3/2} \otimes  {dy}^{3/2}$ for some holomorphic functions $a,b$. The superscript $^{+}$ restricts in each case to sections that are  invariant under the involution $(x,y) \to (y,x)$. The five remaining objects in the diagram can be deduced immediately from these local forms and the obvious inclusions between them.

Our Theorem follows immediately from this Lemma: it suffices to show that the restriction 
\[
\omega_{|C} \in H^1(C, \Hom(\exterior{2}T_-,T_+))
\]
does not vanish. So we hold the even spin curve $(C,T_C^{1/2})$ fixed at some generic value for which $H^0(C,T_C^{1/2})=0,$ and vary $p \in C$. The restriction to $C$ of the odd cotangent sequence \eqref{rel odd cotan} of $\pi$ is then the projection $p_{1*}$ of the SES   {\eqref{SES}}$_{0,3,1}$ of sheaves 
on $C\times C$:
\begin{equation}\label{031}
0 \to \OO(0,3,0) \to \OO(0,3,1) \to  K_C{^{1/2}} \to 0.
\end{equation}
Its extension class in $H^1( K^{-1/2}_{C}) \otimes  H^0(K_{C}^{3/2}) = \End( H^0(K_{C}^{3/2})$) is the identity element, so it is non-zero.
This shows that sequence  \eqref{rel odd cotan}, and hence also sequence $Q$, do not split. We deduce the non-splitting of SES \eqref{with p} by a diagram chase:

Any splitting:
\[
 H^0(K_{C}^{2}(p))
\to H^0(C{\times}C, \mbox{  }  I_p(3(p),3(p),1) \mbox{ } )^{+}
\]
of SES \eqref{with p} restricts to a homomorphism
\[
 H^0(K_{C}^{2})
\to H^0(C{\times}C, \mbox{  }  I_p(3(p),3(p),1) \mbox{ } )^{+}.
\]
We compose this with the given map to $Q_2$. The further image in $Q_3$ vanishes, so it factors through a homomorphism
\[
 H^0(K_{C}^{2})
\to Q_1 = K_C^{1/2} \otimes H^0(K_{C}^{3/2}),
\]
and our assumption that $ H^0(K_C^{1/2})=0$ implies this vanishes. We conclude that any splitting of SES \eqref{with p} restricts to a splitting of SES \eqref{without p} and therefore induces a splitting of the quotient sequence $Q$, but we have just seen that this sequence does not split when restricted to generic even $(C,T_C^{1/2})$.

\end{proof}

\section{Acknowledgments}
We are grateful to Beilinson, Pierre Deligne, Dennis Gaitsgory, Ezra Getzler, Sheldon Katz, Dimitry Leites, Yuri Manin, Rafe  Mazzeo, Tony Pantev, and Albert Schwarz for helpful discussions.
RD acknowledges partial support by NSF grants DMS 1304962 and RTG 0636606. 
EW acknowledges partial support by NSF Grant PHY-1314311.

\end{document}